\documentclass[fleqn,10pt]{wlscirep}
\usepackage[utf8]{inputenc}
\usepackage[T1]{fontenc}

\title{Reducing Circuit Resources in Grover's Algorithm via Constraint-Aware Initialization}

\author[1]{Eunok Bae}
\author[2,3]{Jeonghyeon Shin}
\author[4,*]{Minjin Choi}
\affil[1]{Electronics and Telecommunications Research Institute, Daejeon 34129, Republic of Korea}
\affil[2]{Center for Quantum Technology, Korea Institute of Science and Technology, Seoul 02792, Republic of Korea}
\affil[3]{Department of Mathematics and Research Institute for Basic Sciences, Kyung Hee University, Seoul 02447, Republic of Korea}
\affil[4]{Center for Quantum Information R\&D, Korea Institute of Science and Technology Information, Daejeon 34141, Republic of Korea}

\affil[*]{mathcmj89@gmail.com}

\begin{abstract}
Grover's search algorithm provides a quadratic speedup over classical brute-force methods for unstructured search problems in terms of query complexity and is widely used as a versatile subroutine in numerous quantum algorithms, including those for combinatorial problems with large search spaces.
For such problems, it is natural to reduce the effective search space by incorporating problem constraints at the initialization step, which in Grover's algorithm can be achieved by preparing structured initial states that encode constraint information.
In this work, we present a systematic framework with a simple preprocessing procedure for constraint-aware initialization in Grover's algorithm, focusing on problems with linear constraints.
The proposed framework incorporates cardinality constraints and parity constraints extracted from more general linear constraints.
While such structured initial states can reduce the number of oracle queries required to obtain a solution, their preparation incurs additional circuit-level costs.
We therefore offer a conservative circuit-level resource analysis, showing that the resulting constraint-aware initialization can improve resource efficiency in terms of gate counts and circuit depth.
The proposed framework is further illustrated numerically using the exact-cover problem as a representative application.
Overall, our results indicate that this approach serves as a practical baseline for achieving more resource-efficient implementations of Grover's algorithm compared to the standard uniform initialization.
\end{abstract}
\begin{document}

\flushbottom
\maketitle
%
%
\thispagestyle{empty}

\section*{Introduction}
\label{sec:introduction}

Grover's search algorithm is one of the most prominent quantum algorithms, offering a quadratic speedup over brute-force methods for unstructured search problems~\cite{Grover1996}. 
Beyond its original formulation, it has been widely recognized as a versatile component across a broad range of quantum algorithms.
In particular, Grover's algorithm underlies amplitude amplification~\cite{Brassard2000, Montanaro2016} and has been adapted for quantum counting~\cite{Brassard1998, Aaronson2020}.
It has also been extended from the minimum finding algorithm~\cite{Durr1996} to more general iterative search schemes for optimization, such as Grover adaptive search~\cite{Baritompa2005, Gilliam2021}.
Recent work has explored integrating Grover's algorithm with quantum machine learning, for example in robotic kinematic optimization, and has reported potential performance gains on practical high-dimensional instances~\cite{Nigatu2025}.
This broad applicability highlights the role of Grover's algorithm not only as a breakthrough result in search problems but also as a foundational tool in quantum algorithm design.

Constraint satisfaction problems (CSPs), or more generally, combinatorial optimization problems (COPs), have been actively studied as natural targets for quantum computing approaches, with numerous results reported in recent years~\cite{Campbell2019, Gilliam2021, Ohno2024, Sano2024, Hadfield2019, Bartschi2020, Vikstaal2020, Zhang2025}.
Since constraints are intrinsic to these problems, they are naturally incorporated into the computational framework.
In Grover-based approaches, constraints are typically enforced through constraint-checking operations within the oracle~\cite{Campbell2019, Gilliam2021, Ohno2024}, while variational approaches commonly incorporate them into the problem Hamiltonian through penalty terms~\cite{Vikstaal2020, Bengtsson2020}.
Another complementary approach is to encode available constraint information already in the initial state, thereby restricting the search to a smaller or more relevant subspace.
Such initialization strategies have been explored in both Grover-based approaches and variational algorithms.
For example, several Grover-based studies have investigated problem-tailored initialization for clique search and quadratic assignment problems~\cite{Metwalli2020, Mikuriya2024}, a quantum walk-based reduced initial state for scheduling problems~\cite{Griset2025}, and branching-based iterative state preparation for general linear constraints~\cite{Wilkening2025}.
In variational settings, related approaches include constraint-preserving mixers and ans\"atze~\cite{Hadfield2019, Bartschi2020, Cook2020}, as well as tree-based state generation for the knapsack problem~\cite{Christiansen2025}.
These studies demonstrate that incorporating available constraint information during initialization can effectively reduce the search space before the main computational stage.

In this work, we present a systematic framework for constraint-aware initialization in Grover's algorithm applied to combinatorial problems with linear constraints.
While previous studies have investigated structured initialization for a variety of combinatorial optimization problems, our framework systematically constructs structured initial states from the underlying linear constraints through a greedy preprocessing procedure.
When multiple constraints are available, constructing a structured initial state that simultaneously incorporates them is generally nontrivial because constraints may overlap or require different forms of state preparation.
To address this issue, the preprocessing procedure selects a collection of jointly implementable constraint sets, covering cardinality constraints realized by Dicke states~\cite{Dicke1954} as well as parity information extracted from more general linear constraints. 
We incorporate this parity information via constructions based on Greenberger--Horne--Zeilinger~(GHZ) states~\cite{Greenberger1989} expressed in the $X$ basis (hereafter referred to as GHZ-type states).

However, implementing constraint-aware initialization incurs additional circuit-level overhead, which motivates a careful resource analysis.
We therefore provide a conservative circuit-level resource analysis of the proposed constraint-aware initialization.
In this analysis, we focus on representative instances of our framework, namely the Dicke state with Hamming weight one and GHZ-type states, which capture the essential scaling behavior of cardinality and parity constraints.
In particular, we examine how the number and size of disjoint constraint sets incorporated into the initial state affect gate counts and circuit depth, and show that the overall resource efficiency improves as additional disjoint sets are included. 
Even in the worst-case scenario where the preprocessing ultimately retains only a single constraint set, the resulting initialization strategy remains more resource-efficient than the standard uniform initialization.
We further illustrate the proposed constraint-aware initialization scheme through numerical simulations on the exact-cover problem as a representative application.
The simulation results are consistent with our theoretical analysis, thereby providing numerical support for the proposed framework.

This paper is organized as follows.
The Methods section introduces the proposed framework for constraint-aware initialization in Grover's algorithm.
The theoretical analysis and numerical simulation results are presented and discussed in the Results and Discussion section.
Finally, the Conclusion section summarizes the main findings and outlines possible directions for future work.

\section*{Methods}
\label{sec:methods}

\subsection*{Grover's algorithm with constraint-aware initialization}
\label{subsec:methods01}

We consider the problem of identifying an unknown binary string $x\in \{0, 1\}^{n}$, where $n \in \mathbb{N}$, that satisfies a given condition, typically encoded as a Boolean function $f: \{0, 1\}^{n} \rightarrow \{0, 1\}$.
The goal is to find a binary string $x$ such that $f(x)=1$.

Let $F \subseteq \{0, 1\}^{n}$ denote the search space.
In the standard Grover's algorithm, $F$ corresponds to the full search space $\{0, 1\}^{n}$. 
For the constraint-aware case considered in this work, $F$ denotes a reduced subspace obtained by incorporating some or all of the problem's constraints, such that the subspace $F$ includes the set of all solutions $S$.
The initial state is then prepared as a uniform superposition over all binary strings in $F$ as 
\begin{eqnarray}
\ket{\psi_{F}^{(0)}}_{\mathbf{q}_{[n]}}&=& V_{F}\ket{00 \cdots 0}_{\mathbf{q}_{[n]}}\nonumber \\
&=&\frac{1}{\sqrt{|F|}}\sum_{x \in F}\ket{x}_{\mathbf{q}_{[n]}},
\end{eqnarray}
where $\mathbf{q}_{[n]}=q_{1}q_{2} \cdots q_{n}$ and $V_{F}$ is a unitary operator that generates the desired superposition.
For the full search space $F=\{0, 1\}^{n}$, $V_{F}$ reduces to $H_{q_{1}} \otimes H_{q_{2}} \otimes \cdots \otimes H_{q_{n}}$, where $H$ is the Hadamard operator.
We define the solution and non-solution states as
\begin{equation}
\ket{\phi_{s}}_{\mathbf{q}_{[n]}}=\frac{1}{\sqrt{|S|}}\sum_{x\in S}\ket{x}_{\mathbf{q}_{[n]}} 
\end{equation}
and
\begin{equation}
\ket{\phi_{ns, F}}_{\mathbf{q}_{[n]}}=\frac{1}{\sqrt{|F|-|S|}}\sum_{x \in F\setminus S}\ket{x}_{\mathbf{q}_{[n]}},
\end{equation}
respectively.
Then, the initial state can be expressed as a linear combination of these two orthonormal components as
\begin{equation}
\ket{\psi_{F}^{(0)}}_{\mathbf{q}_{[n]}}=\sin \theta \ket{\phi_{s}}_{\mathbf{q}_{[n]}} + \cos \theta \ket{\phi_{ns, F}}_{\mathbf{q}_{[n]}},
\end{equation}
where $\theta=\arcsin(\sqrt{|S|/|F|})$.

The overall procedure of Grover's algorithm consists of preparing the initial state, repeatedly applying the oracle and diffusion operators for an appropriate number of iterations, and finally measuring the state in the computational basis, as illustrated in Figure~\ref{Figure01}.
The oracle operator $O_{f}$
marks the solutions by flipping their phase:
\begin{equation}
O_{f} \ket{x}_{\mathbf{q}_{[n]}}=(-1)^{f(x)}\ket{x}_{\mathbf{q}_{[n]}}.
\end{equation}
The diffusion operator is defined as a reflection about the initial state:
\begin{equation}
R_{F} = 2\ket{\psi_{F}^{(0)}}\bra{\psi_{F}^{(0)}}_{\mathbf{q}_{[n]}} - \mathbf{I}_{\mathbf{q}_{[n]}},
\end{equation}
where $\mathbf{I}$ is the identity operator.
We refer to the process of applying the oracle operator followed by the diffusion operator as a query.
After performing $\kappa$ queries starting from the initial state, the quantum state becomes
\begin{align}
\ket{\psi_{F}^{(\kappa)}}_{\mathbf{q}_{[n]}}
=&\left(R_{F} \cdot O_{f}\right)^{\kappa}\ket{\psi_{F}^{(0)}}_{\mathbf{q}_{[n]}} \nonumber \\
=& \sin ((2\kappa+1)\theta) \ket{\phi_{s}}_{\mathbf{q}_{[n]}}+ \cos ((2\kappa+1)\theta) \ket{\phi_{ns,F}}_{\mathbf{q}_{[n]}}.
\end{align}
To maximize the amplitude of the solution state $\ket{\phi_{s}}$, $\kappa$ is chosen such that $(2\kappa+1)\theta \approx \pi/2$.
As $\theta=\arcsin(\sqrt{|S|/|F|})$, we define the optimal number of queries as
\begin{equation}
\label{eq:optimal_query}
\kappa_{F}^{opt} = \operatorname{round}\left(\frac{\pi}{4\arcsin(\sqrt{|S|/|F|})} - \frac{1}{2}\right).
\end{equation}
In other words, after $\kappa_{F}^{opt}$ queries, measuring the final state in the computational basis yields a solution with high probability.
We note that when the number of solutions $|S|$ is unknown, the optimal number of queries can be estimated using techniques such as quantum counting~\cite{Brassard1998, Aaronson2020}.
Throughout this work, we assume that $|S|$ is known when evaluating the optimal query number so that the impact of constraint-aware initialization on resource efficiency can be examined independently.

In many practical applications of Grover-based approaches, the oracle is explicitly constructed for the target problem and incorporates problem constraints through constraint-checking operations~\cite{Gilliam2021, Campbell2019, Ohno2024, Jiang2023}.
In such settings, part of the constraint information can be shifted from the oracle to the initialization stage by preparing an initial state that already satisfies selected constraints.
This transfer may increase the cost of state preparation while potentially reducing the complexity of the corresponding oracle construction.
The present work investigates this trade-off through constraint-aware initialization and analyzes the resulting circuit-level resource requirements.

\begin{figure}[t]
\centering
\includegraphics[width=0.6\textwidth]{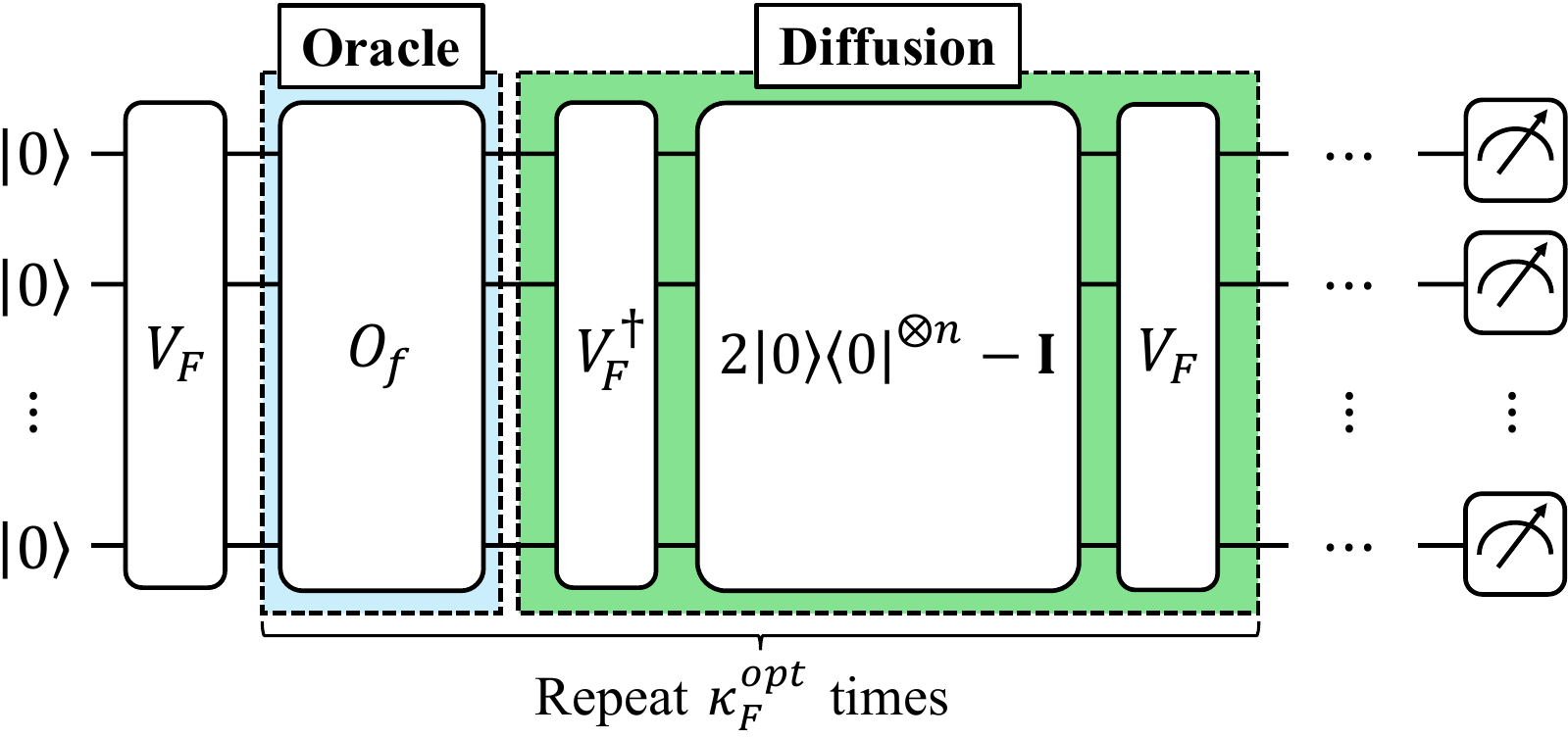}
\caption{
Schematic of Grover's algorithm with constraint-aware initialization.
The initial state $\ket{\psi_{F}^{(0)}}=V_{F}\ket{0}^{\otimes n}$ is a uniform superposition over the search subspace $F$, which may encode some or all of the problem's constraints.
Each query consists of applying the oracle operator $O_{f}$ followed by the diffusion operator $R_{F}=V_{F}(2\ket{0}\bra{0}^{\otimes n} - \mathbf{I})V_{F}^{\dagger}$.
After $\kappa_{F}^{opt}$ queries, a solution can be obtained with high probability upon measurement.}
\label{Figure01}
\end{figure}

\subsection*{Constraint-aware initialization for cardinality constraints}
\label{subsec:methods02}

We first consider a class of linear constraints frequently encountered in combinatorial problems, known as cardinality constraints.
Such constraints restrict the Hamming weight (the number of ones) of a subset of binary variables and can be written as
\begin{equation}
\label{eq:cardinality_constraint}
\sum_{i \in C_{j}}x_{i}=b_{j}, \quad j=1, \dots, m,
\end{equation}
where $C_{j} \subseteq \{1, \dots, n\}$ denotes the subset of variable indices associated with the $j$-th constraint, $x_{i} \in \{0, 1\}$, and $b_{j} \in \mathbb{N}$.
This simple yet important form arises in a variety of combinatorial problems, including the exact-cover problem~\cite{Vikstaal2020, Bengtsson2020, Jiang2023} and the traveling salesman problem~\cite{Bartschi2020, Sato2023, Sato2025}.
While the standard implementation employs a uniform superposition over all $2^{n}$ binary strings as the initial state, we instead incorporate some of the constraints in Eq.~(\ref{eq:cardinality_constraint}) directly into the initial state, so that the search space is reduced to the set of basis states satisfying the selected constraints.
We describe how this constraint-aware initialization can be realized using structured quantum states.

To illustrate the idea, let us examine a single cardinality constraint, which requires that the sum of the binary variables in a subset $C\subseteq \{1, \dots, n\}$ equals $b$.
This constraint enforces that exactly $b$ of the variables indexed by $C$ are set to one.
A natural way to reflect this constraint in the initial state is to prepare a Dicke state $\ket{D^{|C|}_{b}}$ on the corresponding qubits, where $|C|$ is the size of the subset $C$.
A Dicke state $\ket{D_{\nu}^{\mu}}$ is defined as the equal superposition of all computational basis states of length $\mu$ with Hamming weight exactly $\nu$:
\begin{equation}
\label{eq:Dicke}
\ket{D_{\nu}^{\mu}}
= \frac{1}{\sqrt{\binom{\mu}{\nu}}}\sum_{\substack{x \in \{0, 1\}^{\mu} \\ \mathrm{wt}(x)=\nu}}\ket{x},
\end{equation}
where $\mathrm{wt}(x)$ is the Hamming weight of the binary string $x$.
Dicke states can be prepared deterministically without ancillary qubits using the method proposed by B\"{a}rtschi and Eidenbenz~\cite{Bartschi2019}, which requires $O(\mu\nu)$ gates and $O(\mu)$ circuit depth.
This indicates that employing Dicke states for constraint-aware initialization can be achieved with reasonable circuit resources.

Given a set of cardinality constraints of the form in Eq.~(\ref{eq:cardinality_constraint}), one can in principle construct the initial state by preparing independent Dicke states on disjoint subsets of qubits corresponding to those constraints, while applying Hadamard gates to the remaining qubits.
If a collection of disjoint subsets $\{C'_{1}, \dots, C'_{t}\}$ can be identified, the resulting initial state can be written, after an appropriate reordering of qubit indices, as
\begin{equation}
\label{eq:init_Dicke_multi_disj}
\ket{\psi_{F}^{(0)}}_{\mathbf{q}_{[n]}} = \ket{D_{b'_{1}}^{|C'_{1}|}}_{\mathbf{q}_{C'_{1}}} \otimes \cdots \otimes \ket{D_{b'_{t}}^{|C'_{t}|}}_{\mathbf{q}_{C'_{t}}}
\otimes H\ket{0}_{q_{\gamma_{t}+1}} \otimes \cdots \otimes  H\ket{0}_{q_{n}},
\end{equation}
where $\mathbf{q}_{\Lambda}$ represents the set of qubits associated with the indices in a subset $\Lambda$, and $\gamma_{t}=\sum_{j=1}^{t}|C'_{j}|$ is the total number of qubits involved in the preparation of Dicke states.
This construction reduces the search space size from $2^{n}$ to $2^{n-\gamma_{t}} \cdot \prod_{j=1}^{t}\binom{|C'_{j}|}{b'_{j}}$.
In practice, identifying an optimal collection of disjoint subsets $\{C'_{j}\}$ that maximizes this reduction is itself a combinatorial problem and can be nontrivial.
Nevertheless, a simple greedy procedure in Algorithm~\ref{alg:preprocessing1} can efficiently produce a feasible disjoint subset structure.
The constraints are first sorted in ascending order of $\binom{|C_j|}{b_{j}}/2^{|C_{j}|}$, so that those imposing the strongest restriction on the feasible search space are prioritized.
The algorithm then iteratively selects non-overlapping subsets, corresponding to the case where the overlap threshold is set to $\eta_{\rm{th}}=0$.
This preprocessing step requires a classical cost of $O(mk)$, in addition to an initial sorting cost of $O(m\log m)$, where $m$ is the number of constraints and $k$ denotes the typical subset size.

\begin{algorithm}[t]
\caption{Classical preprocessing for constraint-aware initialization with cardinality constraints}
\label{alg:preprocessing1}
\begin{algorithmic}[1]
\State \textbf{Input:} Collection of pairs $\{(C_{j}, b_{j})\}_{j=1}^{m}$; overlap threshold $\eta_{\rm{th}}$
\State \textbf{Output:} Selected subset collection $\mathcal{T}$
\State $\mathcal{T} \gets \emptyset$, $\mathcal{C}_{\text{used}} \gets \emptyset$
\State $\mathcal{L} \gets \{C_{j}\}_{j=1}^{m}$ sorted in ascending order of $\binom{|C_j|}{b_{j}}/2^{|C_{j}|}$

\vspace{4pt}
\Statex \textbf{Phase 1 (disjoint)}
\For{\textbf{each} $C_{j}$ in $\mathcal{L}$}
    \If{$C_{j} \cap \mathcal{C}_{\text{used}} = \emptyset$}
        \State $\mathcal{T} \gets \mathcal{T} \cup \{C_{j}\}$
        \State $\mathcal{C}_{\text{used}} \gets \mathcal{C}_{\text{used}} \cup C_{j}$
    \EndIf
\EndFor

\vspace{4pt}
\Statex \textbf{Phase 2 (overlap)}
\If{$\eta_{\rm{th}}>0$}
    \State $\mathcal{R} \gets \{C_{j} \in \mathcal{L}:~C_{j} \notin \mathcal{T}\}$
    \For{\textbf{each} $C_{j}$ in $\mathcal{R}$}
        \State $r \gets |C_j \cap \mathcal{C}_{\text{used}}|$
        \If{$r \le \eta_{\rm{th}}$}
            \State $R_{j} \gets C_{j} \setminus \mathcal{C}_{\text{used}}$
            \If{$|R_{j}|>0$}
                \State $\mathcal{T} \gets \mathcal{T} \cup \{R_{j}\}$
                \State $\mathcal{C}_{\text{used}} \gets \mathcal{C}_{\text{used}} \cup R_{j}$
            \EndIf
        \EndIf
    \EndFor
\EndIf    
\State \Return $\mathcal{T}$
\end{algorithmic}
\end{algorithm}

The above construction considers only the case where the selected subsets are mutually disjoint.
In many practical settings, however, constraints may share a number of common variables.
To handle such cases, the greedy preprocessing can be extended to allow small overlaps among subsets.
As described in the second phase of Algorithm~\ref{alg:preprocessing1}, subsets whose overlap with the already selected ones is below a chosen threshold $\eta_{\rm{th}}$ are sequentially included after removing the shared elements.
This additional phase involves only local overlap checks between subsets and thus requires the same classical cost of $O(mk)$ as in the disjoint case.
Here, we denote by $\mathcal{C}_{\text{used}}$ the set of variable indices that have already been included in previously selected subsets.
For a constraint subset $C_{j}$ overlapping with $\mathcal{C}_{\text{used}}$ by $r_{j}=|C_{j} \cap \mathcal{C}_{\text{used}}| \le \eta_{\rm{th}}$, we define the reduced set $R_{j}=C_{j}\setminus \mathcal{C}_{\text{used}}$ and relax the corresponding cardinality constraint to
\begin{equation}
\label{eq:relaxed_constraint}
\max(0, b_{j}-r_{j}) \le \sum_{i \in R_{j}}x_{i} \le b_{j}.
\end{equation}
This relaxed form defines a valid subspace over the remaining variables without conflicting with the original constraint.
The corresponding quantum state is a superposition of Dicke states over the allowed Hamming weights,
\begin{equation}
\label{eq:relaxed_Dicke}
\ket{D_{b_{j}, r_{j}}^{|R_{j}|}}=\frac{1}{\sqrt{\Omega_{j}}}\sum_{\nu=b_{j}^{\rm{min}}}^{b_{j}} \sqrt{\binom{|R_{j}|}{\nu}}\ket{D_{\nu}^{|R_{j}|}},
\end{equation}
where $b_{j}^{\rm{min}}=\max(0, b_{j}-r_{j})$ and $\Omega_{j}=\sum_{\nu=b_{j}^{\rm{min}}}^{b_{j}}\binom{|R_{j}|}{\nu}$.
Preparing a linear combination of Dicke states of the form $\sum_{\nu}\alpha_{\nu}\ket{D_{\nu}^{\mu}}$ can be implemented with $O(\mu^{2})$ gates and $O(\mu)$ depth~\cite{Bartschi2019}.

In summary, after preprocessing the constraints and reordering the qubit indices, the initial state can be constructed in the product form
\begin{eqnarray}
\label{eq:init_Dicke_multi_overl}
\ket{\psi_{F}^{(0)}}_{\mathbf{q}_{[n]}} = \ket{D_{b'_{1}}^{|C'_{1}|}}_{\mathbf{q}_{C'_{1}}} \otimes \cdots \otimes \ket{D_{b'_{t}}^{|C'_{t}|}}_{\mathbf{q}_{C'_{t}}} \otimes \ket{D_{b'_{t+1}, r'_{t+1}}^{|R'_{t+1}|}}_{\mathbf{q}_{R'_{t+1}}} \otimes \cdots \otimes \ket{D_{b'_{s}, r'_{s}}^{|R'_{s}|}}_{\mathbf{q}_{R'_{s}}}
\otimes H\ket{0}_{q_{\gamma_{s}+1}} \otimes \cdots \otimes  H\ket{0}_{q_{n}},
\end{eqnarray}
where $\gamma_{s}$ is the total number of qubits involved in the preparation of Dicke states and their linear combinations.
As a result, the effective search space size is reduced from $2^{n}$ to $2^{n-\gamma_{s}} \cdot \prod_{j=1}^{t}\binom{|C'_{j}|}{b'_{j}} \cdot \prod_{j=t+1}^{s}\Omega_{j}$.
Although more general, non-product forms of initialization are conceivable, the product structure in Eq.~(\ref{eq:init_Dicke_multi_overl}) offers a practical advantage.
It provides a systematic and modular framework for embedding multiple constraints while keeping the circuit depth proportional to the largest subset size.
Consequently, this approach enables constraint-aware initialization with relatively shallow depth while encoding some of the constraints, and can reduce both the effective search space and the overall circuit-level resource requirements of Grover's algorithm, as will be demonstrated in the Results and Discussion section.

\subsection*{Constraint-aware initialization for parity constraints extracted from linear constraints}
\label{subsec:methods03}
We next examine a more general linear constraint of the form
\begin{equation}
\label{eq:general_constraint_single}
\sum_{i \in C}\alpha_{i}x_{i}=b,
\end{equation}
where $C \subseteq \{1, \dots, n\}$, $x_{i} \in \{0, 1\}$, and $\alpha_{i}, b \in \mathbb{Z}$.
Without loss of generality, we assume that the greatest common divisor of all coefficients $\{\alpha_{i}\}_{i \in C}$ and $b$ is one.
In contrast to the cardinality constraint discussed in the previous subsection, the number of ones among $\{x_{i}\}_{i \in C}$ cannot be predetermined in this case.
Thus, the method based on the Dicke state is not applicable.
To construct a feasible initialization in this setting, we instead exploit the parity structure of the coefficients.
Let $\hat{C}=\{i\in C|\alpha_{i}~\text{is odd}\}$. 
By reducing Eq.~(\ref{eq:general_constraint_single}) modulo 2, we obtain
\begin{equation}
\label{eq:general_constraint_single_reduced}
\sum_{i \in \hat{C}}x_{i} \equiv b \pmod{2},
\end{equation}
which implies that if $b$ is even (odd), then the number of ones among $\{x_{i}\}_{i \in \hat{C}}$ must also be even (odd).
This parity condition can be encoded using a GHZ-type entangled state of the form
\begin{eqnarray}
\label{eq:GHZ}
\ket{GHZ_{\mu,\nu}^{(X)}}
&=&\frac{1}{\sqrt{2}}\left(\ket{+}^{\otimes \mu} + (-1)^{\nu}\ket{-}^{\otimes \mu}\right) \nonumber \\
&=&\frac{1}{\sqrt{2^{\mu-1}}}\sum_{\substack{x \in \{0, 1\}^{\mu} \\ \mathrm{wt}(x)\equiv \nu \pmod{2}}}\ket{x},
\end{eqnarray}
where $\ket{\pm}=\frac{1}{\sqrt{2}}(\ket{0}\pm\ket{1})$, $\mu \in \mathbb{N}$, and $\nu \in \{0, 1\}$.
In the present setting, we set $\mu=|\hat{C}|$ and take $\nu=b \pmod{2}$, reflecting the parity condition in Eq.~(\ref{eq:general_constraint_single_reduced}).

Now suppose that a set of linear constraints of the form
\begin{equation}
\label{eq:general_constraint_multiple}
\sum_{i \in C_{j}}\alpha_{j,i}x_{i}=b_{j}, \quad j=1, \dots, m,
\end{equation}
is given, where, for each $j$, the coefficients $\{\alpha_{j,i}\}$ and $b_{j}$ have greatest common divisor one.
Each constraint induces a parity condition on the subset $\hat{C}_{j}=\{i\in C_{j}|\alpha_{j,i}~\text{is odd}\}$.
If one can identify disjoint subsets $\{\hat{C}'_{1}, \dots, \hat{C}'_{t}\}$, then, after reordering the qubit indices, the initial state can be written as
\begin{eqnarray}
\label{eq:init_GHZ_multi_disj}
\ket{\psi_{F}^{(0)}}_{\mathbf{q}_{[n]}} = \ket{GHZ_{|\hat{C}'_{1}|, b'_{1}}^{(X)}}_{\mathbf{q}_{\hat{C}'_{1}}} \otimes \cdots \otimes \ket{GHZ_{|\hat{C}'_{t}|, b'_{t}}^{(X)}}_{\mathbf{q}_{\hat{C}'_{t}}} \otimes H\ket{0}_{q_{\hat{\gamma}_{t}+1}} \otimes \cdots \otimes  H\ket{0}_{q_{n}},
\end{eqnarray}
where $\hat{\gamma}_{t}=\sum_{j=1}^{t}|\hat{C}'_{j}|$ denotes the total number of qubits involved in the parity encoding.
Since each constraint effectively reduces the search space by half, the resulting search space size becomes $2^{n-t}$, independent of the subset sizes $|\hat{C}_{j}|$.
To maximize the overall reduction, it is therefore desirable to identify as many disjoint subsets as possible.
To this end, we employ a simple greedy preprocessing procedure that iteratively selects disjoint parity subsets in ascending order of their sizes, as described in Algorithm~\ref{alg:preprocessing2}.

\begin{algorithm}[t]
\caption{Classical preprocessing for constraint-aware initialization using parity constraints extracted from linear constraints}
\label{alg:preprocessing2}
\begin{algorithmic}[1]
\State \textbf{Input:} Collection of triples $\{(C_{j}, \{\alpha_{j,i}\}_{i \in C_{j}},b_{j})\}_{j=1}^{m}$
\State \textbf{Output:} Selected subset collection $\mathcal{T}$
\State $\mathcal{T} \gets \emptyset$, $\mathcal{C}_{\text{used}} \gets \emptyset$
\For{$j=1, \dots, m$}
    \State $\hat{C}_{j} \leftarrow \{\,i \in C_{j}: \alpha_{j, i} \text{ is odd}\,\}$
\EndFor
\State $\mathcal{L} \gets \{\hat{C}_{j}\}_{j=1}^{m}$ sorted in ascending order of $|\hat{C}_{j}|$
\vspace{4pt}
\Statex \textbf{(Disjoint subset selection)}
\For{\textbf{each} $\hat{C}_{j}$ in $\mathcal{L}$}
    \If{$\hat{C}_{j} \cap \mathcal{C}_{\text{used}} = \emptyset$}
        \State $\mathcal{T} \gets \mathcal{T} \cup \{\hat{C}_{j}\}$
        \State $\mathcal{C}_{\text{used}} \gets \mathcal{C}_{\text{used}} \cup \hat{C}_{j}$
    \EndIf
\EndFor

\State \Return $\mathcal{T}$
\end{algorithmic}
\end{algorithm}

When both cardinality and parity constraints are present, the initialization can be systematically constructed by combining Dicke states and GHZ-type states.
Since encoding a cardinality constraint on a subset $C_{j}$ via Dicke state preparation reduces the search space by a factor of $\binom{|C_{j}|}{b_{j}}/2^{|C_{j}|}$, whereas each parity constraint reduces it by exactly one half under the proposed initialization scheme, it can be advantageous in terms of search space reduction to prioritize subsets associated with cardinality constraints before applying the parity-based selection to the remaining constraints.
This ordering enables an efficient reduction of the overall search space within a systematic framework.

\section*{Results and Discussion}
\label{sec:results}

\subsection*{Resource requirement analysis}
\label{subsec:results01}

We evaluate the overall resource requirements associated with the constraint-aware initialization considered in this work.
In particular, we examine how structural features of the imposed constraints, such as the number of constraints incorporated into the initialization and the size of the constraint sets, affect the resulting resource requirements.
The resource metrics considered here include the total gate count, the circuit depth, and the number of two-qubit gates.
The total gate count serves as a measure of circuit complexity, the circuit depth reflects the computation time, and the number of two-qubit gates provides a practical indicator of error sensitivity in experimental implementations.

Throughout this resource analysis, we focus on two representative cases, namely the states $\ket{D_{1}^{\mu}}$ and $\ket{GHZ_{\mu,\nu}^{(X)}}$, which encode cardinality and parity constraints, respectively.
Although our quantitative analysis centers on these representative examples, the underlying arguments rely only on generic features such as the reduction of the effective search space and the overhead associated with state preparation.
As a result, while the explicit form of the efficiency condition may depend on the specific structure of the initial state preparation, we expect that similar criteria and analytical frameworks can be employed to assess the relative efficiency of more general forms of constraint-aware initialization.

For an initialization strategy $\sigma$, let $\mathcal{S}_{\sigma}$ denote the cost of preparing the initial state, and $\mathcal{O}_{\sigma}$ the cost of implementing the oracle. 
Each strategy $\sigma$ induces a search space $F_{\sigma}$, and the corresponding optimal number of queries $\kappa_{F_{\sigma}}^{opt}$ defined in Eq.~(\ref{eq:optimal_query}) is abbreviated as $\kappa_{\sigma}$.
Let $\mathcal{D}$ denote the cost of implementing the unitary $2\ket{0}\bra{0}^{\otimes n} - \mathbf{I}$, which appears in the diffusion operator and is independent of the choice of initialization strategy.
Assuming that implementing a unitary and its inverse incurs the same cost, the total resource requirement can be expressed as
\begin{equation}
\label{eq:eq_resource}
\mathcal{R}_{\sigma}=\mathcal{S}_{\sigma}+\left(\mathcal{O}_{\sigma}+\mathcal{D}+2\mathcal{S}_{\sigma}\right)\kappa_{\sigma},
\end{equation}
where the first term corresponds to the one-time cost of state preparation, while the second term accounts for repeated applications of the oracle and both the forward and inverse state preparation operators in each query.
In practice, the overall cost may be further reduced through optimizations such as commuting operations.
However, adopting the expression in Eq.~(\ref{eq:eq_resource}) provides a consistent and transparent basis for comparing the relative efficiency of different initialization strategies.
We say that strategy $\tau$ is more efficient than strategy $\sigma$ (with respect to a given metric) if $\mathcal{R}_{\tau} < \mathcal{R}_{\sigma}$.

Before analyzing specific initialization schemes, we derive a sufficient condition under which one initialization strategy is more resource-efficient than another.
\begin{Lem}
\label{result:lem1}
For strategies $\sigma$ and $\tau$, suppose that $\mathcal{O}_{\tau} \le \mathcal{O}_{\sigma}$ and $\kappa_{\tau}<\kappa_{\sigma}$.
If
\begin{equation}
\label{eq:ineq_resource}
\mathcal{O}_{\sigma} + \mathcal{D} > \left(\frac{2\kappa_{\tau}+1}{\kappa_{\sigma}-\kappa_{\tau}}\right)\left(\mathcal{S}_{\tau}-\mathcal{S}_{\sigma}\right)-2\mathcal{S}_{\sigma},
\end{equation}
then $\mathcal{R}_{\tau} < \mathcal{R}_{\sigma}$.
\end{Lem}
\noindent\textbf{Proof.}
Since $\mathcal{O}_{\tau} \le \mathcal{O}_{\sigma}$, 
\begin{equation}
\mathcal{R}_{\tau}=\mathcal{S}_{\tau}+\left(\mathcal{O}_{\tau}+\mathcal{D}+2\mathcal{S}_{\tau}\right)\kappa_{\tau} \le \mathcal{S}_{\tau}+\left(\mathcal{O}_{\sigma}+\mathcal{D}+2\mathcal{S}_{\tau}\right)\kappa_{\tau}.
\end{equation}
Hence, whenever
\begin{equation}
\label{eq:modified_condition}
\mathcal{S}_{\tau}+\left(\mathcal{O}_{\sigma}+\mathcal{D}+2\mathcal{S}_{\tau}\right)\kappa_{\tau} < \mathcal{R}_{\sigma},
\end{equation}
it follows that $\mathcal{R}_{\tau} < \mathcal{R}_{\sigma}$.
Substituting $\mathcal{R}_\sigma=\mathcal{S}_{\sigma}+\left(\mathcal{O}_{\sigma}+\mathcal{D}+2\mathcal{S}_{\sigma}\right)\kappa_{\sigma}$ into Eq.~(\ref{eq:modified_condition}) and rearranging gives
\begin{align}
(\mathcal{O}_{\sigma} + \mathcal{D})(\kappa_{\sigma} - \kappa_{\tau}) &>2\mathcal{S}_{\tau}\kappa_{\tau}-2\mathcal{S}_{\sigma}\kappa_{\sigma}+\mathcal{S}_{\tau}-\mathcal{S}_{\sigma} \nonumber \\
&= 2\kappa_{\tau}(\mathcal{S}_{\tau}-\mathcal{S}_{\sigma})-2\mathcal{S}_{\sigma}(\kappa_{\sigma}-\kappa_{\tau})+\mathcal{S}_{\tau}-\mathcal{S}_{\sigma}.
\end{align}
Since $\kappa_{\tau}<\kappa_{\sigma}$, dividing both sides by $\kappa_{\sigma}-\kappa_{\tau}$ yields the criterion in Eq.~(\ref{eq:ineq_resource}).
\hfill $\square$

\subsubsection*{Cardinality constraints}
\label{subsec:results01-01}

Let us first consider the case of cardinality constraints and analyze how the incremental incorporation of disjoint constraint sets influences the overall resource requirement.
To this end, we introduce a sequence of strategies $\sigma_{i}$, where $\sigma_{0}$ corresponds to the standard uniform initialization without constraints, and $\sigma_{i}$ denotes the strategy obtained by incorporating $i$ disjoint constraint sets into the initial state preparation.
Each successive strategy $\sigma_{i+1}$ is constructed by adding one additional disjoint constraint set of size $\mu_{i+1}$ to $\sigma_{i}$, and we derive conditions under which $\sigma_{i+1}$ yields a lower resource cost than $\sigma_{i}$.
When the additional constraint set is incorporated into the initial state, the corresponding constraint no longer needs to be enforced within the oracle.
Hence, the oracle used for $\sigma_{i}$ can be reused for $\sigma_{i+1}$, which implies $\mathcal{O}_{\sigma_{i+1}} \le \mathcal{O}_{\sigma_i}$.
Moreover, if the oracle is constructed by sequentially applying constraint-checking operations for each constraint, then incorporating the additional constraint into the initialization removes the corresponding checking operation from the oracle and may therefore lead to a strict reduction in oracle cost.
The following proposition formalizes the comparison for the case in which the newly added constraint set is prepared as a Dicke state $\ket{D_{1}^{\mu_{i+1}}}$.

\begin{Prop}
\label{result:prop1}
Suppose that strategy $\sigma_{i+1}$ is obtained from $\sigma_{i}$ by incorporating one additional disjoint constraint set represented by $\ket{D_{1}^{\mu}}$, where we simply denote $\mu_{i+1}$ by $\mu\ge 2$.
Suppose further that $\mathcal{O}_{\sigma_{i+1}} \le \mathcal{O}_{\sigma_i}$, $|F_{\sigma_i}| \ge 64|S|$, and $|F_{\sigma_{i+1}}| \ge 2|S|$, where $|S|>0$ is the number of solutions.
If
\begin{equation}
\label{eq:prop1_01}
\mathcal{O}_{\sigma_{i}} + \mathcal{D} >\left(\frac{24\pi}{21\sqrt{\frac{2^{\mu}}{\mu}}-8\pi}\right)\left(\mathcal{S}_{\sigma_{i+1}}-\mathcal{S}_{\sigma_i}\right)-2\mathcal{S}_{\sigma_i},
\end{equation}
then strategy $\sigma_{i+1}$ is more efficient than strategy $\sigma_i$.
In particular, when the Dicke state $\ket{D_{1}^{\mu}}$ is prepared using the circuit shown in Figure~\ref{Figure02}(a), the sufficient condition simplifies to
\begin{equation}
\label{eq:prop1_02}
\mathcal{O}_{\sigma_{i}} + \mathcal{D} \ge 67
\end{equation}
for each of the considered resource metrics (total gate count, circuit depth, and number of two-qubit gates).
\end{Prop}

\begin{figure*}[t]
\centering
\includegraphics[width=0.9\textwidth]{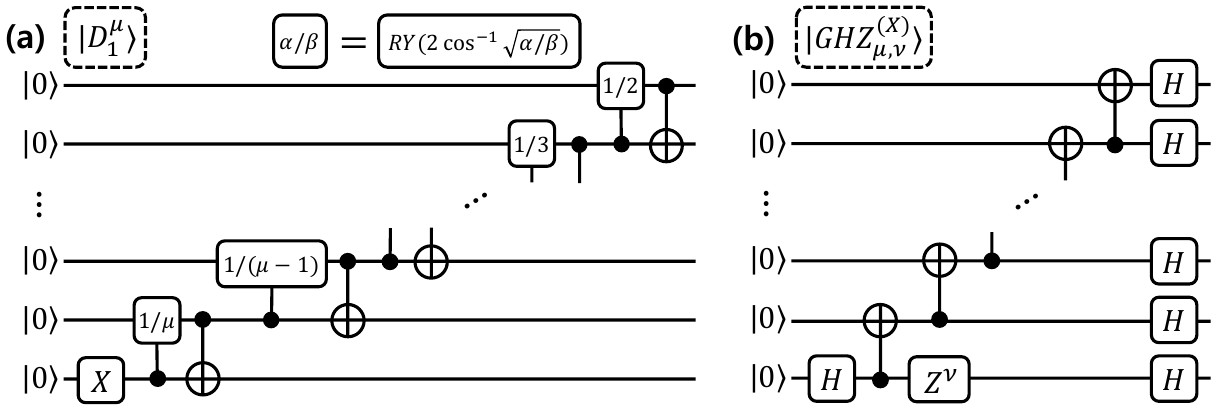}
\caption{
(a) Quantum circuit for preparing $\ket{D_{1}^{\mu}}$ on $\mu$ qubits.
(b) Quantum circuit for preparing $\ket{GHZ_{\mu,\nu}^{(X)}}$ on $\mu$ qubits.
From the circuits, the total gate count, the two-qubit gate count, and the circuit depth are $2\mu-1$, $2\mu-2$, and $2\mu-1$ for (a), and $2\mu+1$, $\mu-1$, and $\mu+1$ for (b), respectively.
}
\label{Figure02}
\end{figure*}

\noindent\textbf{Proof.}
Let $x_{i}=\sqrt{|S|/|F_{\sigma_{i}}|}$ and $x_{i+1}=\sqrt{|S|/|F_{\sigma_{i+1}}|}$.
Then $x_{i+1} \ge \sqrt{2}x_{i}$ for $\mu \ge 2$ because $|F_{\sigma_{i}}|/|F_{\sigma_{i+1}}|=2^{\mu}/\mu$.
Since $\arcsin x$ is strictly convex on $(0, 1)$ with $\arcsin 0 =0$ and $0<x_i\le 1/8$, it follows that
\begin{equation}
\frac{1}{\arcsin x_{i+1}} < \frac{1}{\sqrt{2}\arcsin x_{i}}.
\end{equation}
Using $\kappa_{\sigma_i} \ge \pi/(4\arcsin x_{i})-1$ and $\kappa_{\sigma_{i+1}} \le \pi/(4\arcsin x_{i+1})$, we obtain
\begin{equation}
\kappa_{\sigma_i}-\kappa_{\sigma_{i+1}}
>\frac{\pi}{4\arcsin x_{i}}\left(1-\frac{1}{\sqrt{2}}\right)-1.
\end{equation}
Since $|F_{\sigma_i}| \ge 64|S|$ and $\arcsin x$ is increasing on $[0, 1]$, the right-hand side is positive, which implies $\kappa_{\sigma_i}>\kappa_{\sigma_{i+1}}$.
Together with the assumption $\mathcal{O}_{\sigma_{i+1}} \le \mathcal{O}_{\sigma_i}$, Lemma~\ref{result:lem1} can be applied and it remains to bound
\begin{equation}
\frac{2\kappa_{\sigma_{i+1}}+1}{\kappa_{\sigma_{i}}-\kappa_{\sigma_{i+1}}}.
\end{equation}
Since $x_{i} \le 1/8 < 0.17$, it follows from the bound $\arcsin x \le 8\pi x/25$ for $0 \le x \le 0.17$ that
\begin{equation}
\kappa_{\sigma_{i}} \ge \frac{\pi}{4\arcsin x_{i}}-1 \ge \frac{21}{32x_{i}}.
\end{equation}
On the other hand, since $x \le \arcsin x$ for $0 \le x \le 1$,
\begin{equation}
\kappa_{\sigma_{i+1}} \le \frac{\pi}{4\arcsin x_{i+1}} \le \frac{\pi}{4x_{i+1}}.
\end{equation}
Therefore, using $x_{i+1}/x_{i}=\sqrt{2^{\mu}/\mu}$ and the fact that $|F_{\sigma_{i+1}}| \ge 2|S|$ implies $1 \le \pi/(4x_{i+1})$,
\begin{equation}
\frac{2\kappa_{\sigma_{i+1}}+1}{\kappa_{\sigma_i}-\kappa_{\sigma_{i+1}}} 
\le \frac{\frac{3\pi}{x_{i}}\sqrt{\frac{\mu}{2^{\mu}}}}{4\kappa_{\sigma_i}-\frac{\pi}{x_{i}}\sqrt{\frac{\mu}{2^{\mu}}}}\le \frac{24\pi}{21\sqrt{\frac{2^{\mu}}{\mu}}-8\pi},
\end{equation}
which gives the criterion in Eq.~\eqref{eq:prop1_01}.

Under the specific implementation where the Dicke state $\ket{D_{1}^{\mu}}$ is prepared using the circuit shown in Figure~\ref{Figure02}(a), the additional cost incurred by incorporating one more constraint set satisfies
\begin{equation}
\mathcal{S}_{\sigma_{i+1}} - \mathcal{S}_{\sigma_{i}} \le 2\mu,
\end{equation}
for each of the considered resource metrics, such as the total gate count, the circuit depth, and the number of two-qubit gates.
Therefore, we obtain a sufficient condition for this specific implementation,
\begin{equation}
\mathcal{O}_{\sigma_{i}} + \mathcal{D} >\frac{48\pi\mu}{21\sqrt{\frac{2^\mu}{\mu}}-8\pi},
\end{equation}
where we omit the negative term $-2\mathcal{S}_{\sigma_{i}}$.
It can be shown that the right-hand side decreases monotonically for $\mu \ge 2$, and its value at $\mu=2$ is approximately $66.1$.  
Hence, in this case, the inequality in Eq.~\eqref{eq:prop1_02} guarantees that strategy $\sigma_{i+1}$ is more efficient than strategy $\sigma_i$.
\hfill $\square$

An interesting observation is that the cost $\mathcal{D}$ of implementing $2\ket{0}\bra{0}^{\otimes n} - \mathbf{I}$ alone is typically sufficient to satisfy the inequality in Eq.~(\ref{eq:prop1_02}), even without taking the oracle construction into account.
In practice, $\mathcal{D}$ is dominated by the implementation of the multi-controlled $Z$ gate (which can be expressed using Hadamard gates and a multi-controlled $X$ gate). 
Decomposing this component using a standard ancilla-assisted construction~\cite{Nielsen2010} yields gate counts and circuit depths that readily exceed the threshold value of 67 even for a relatively small number of qubits.
The same conclusion holds for ancilla-free decompositions, which generally require substantially higher resource costs~\cite{Maslov2016}. 
Consequently, in many practical settings, the sufficient condition in Eq.~\eqref{eq:prop1_02} can be satisfied even before accounting for the oracle construction cost.
This suggests that incorporating additional disjoint cardinality constraints into the initialization can improve overall resource efficiency.
Notably, even when only a single constraint set is used, the resulting initialization strategy can outperform the standard approach, highlighting that effectively exploiting readily available constraints can be practically meaningful without exhaustive constraint coverage.

We next study how the size of an additional constraint set affects resource efficiency with $\ket{D_{1}^{\mu}}$.
Unlike the preceding comparison, which compares a strategy that incorporates an additional constraint with one that does not, the comparison considered here concerns two strategies that both incorporate an additional constraint but differ in the size of the additional constraint.
Consequently, the relation between the corresponding oracle costs is more implementation-dependent, and the oracle-reuse argument used in Proposition~\ref{result:prop1} and~\ref{result:prop3} does not directly apply.
For example, for a fixed cardinality value $b$, checking whether a set of variables satisfies $\sum_{i \in C}x_{i}=b$ may require a more costly circuit as the size of $C$ increases.
In this case, encoding the constraint into the initialization removes the corresponding constraint-checking operation from the oracle, so eliminating the checking operation for a larger cardinality set may result in a lower oracle implementation cost.
The following proposition formalizes one such setting.

\begin{Prop}
\label{result:prop2}
Suppose that two strategies $\sigma$ and $\sigma'$ both incorporate an additional disjoint constraint set into the initialization and differ in the size of the additional constraint set, which is represented by $\ket{D_{1}^{\mu}}$ for $\sigma$ and $\ket{D_{1}^{\mu+1}}$ for $\sigma'$, where $\mu \ge 2$.
Suppose further that $\mathcal{O}_{\sigma'} \le \mathcal{O}_{\sigma}$, $|F_{\sigma}| \ge 100|S|$, and $|F_{\sigma'}| \ge 2|S|$, where $|S|>0$ is the number of solutions.
If
\begin{equation}
\label{eq:prop2_01}
\mathcal{O}_{\sigma} + \mathcal{D} > \left(\frac{120\pi}{109\sqrt{\frac{2\mu}{\mu+1}}-40\pi}\right)\left(\mathcal{S}_{\sigma'}-\mathcal{S}_{\sigma}\right)-2\mathcal{S}_{\sigma},
\end{equation}
then strategy $\sigma'$ is more efficient than strategy $\sigma$.
In particular, when the Dicke state $\ket{D_{1}^{\mu}}$ is prepared using the circuit shown in Figure~\ref{Figure02}(a), the sufficient condition simplifies to
\begin{equation}
\label{eq:prop2_02}
\mathcal{O}_{\sigma} + \mathcal{D} \ge 97,
\end{equation}
for $\mu \ge 3$. 
This condition applies to each of the considered resource metrics (total gate count, circuit depth, and number of two-qubit gates).
\end{Prop}

\noindent\textbf{Proof.}
Let $x=\sqrt{|S|/|F_{\sigma}|}$ and $x'=\sqrt{|S|/|F_{\sigma'}|}$.
Then $x' \ge \sqrt{4/3}x$ for $\mu \ge 2$ because $|F_{\sigma}|/|F_{\sigma'}|=2\mu/(\mu+1)$.
Repeating the argument used in Proposition~\ref{result:prop1} under the conditions $\mathcal{O}_{\sigma'} \le \mathcal{O}_{\sigma}$, $|F_{\sigma}| \ge 100|S|$, and $|F_{\sigma'}| \ge 2|S|$, one obtains $\kappa_{\sigma'}<\kappa_{\sigma}$ and
\begin{equation}
    \frac{2\kappa_{\sigma'}+1}{\kappa_{\sigma}-\kappa_{\sigma'}}\le \frac{120\pi}{109\sqrt{\frac{2\mu}{\mu+1}}-40\pi}.
\end{equation}
Applying Lemma~\ref{result:lem1} then yields the criterion in Eq.~\eqref{eq:prop2_01}.

Under the specific implementation where the Dicke state $\ket{D_{1}^{\mu}}$ is prepared using the circuit shown in Figure~\ref{Figure02}(a), we have
\begin{equation}
\mathcal{S}_{\sigma'} - \mathcal{S}_{\sigma} = 2.
\end{equation}
Substituting this relation into Eq.~\eqref{eq:prop2_01} gives
\begin{equation}
\mathcal{O}_{\sigma} + \mathcal{D} > \frac{240\pi}{109\sqrt{\frac{2\mu}{\mu+1}}-40\pi}.
\end{equation}
The right-hand side is decreasing for $\mu \ge 2$, and its value drops sharply when $\mu$ increases from 2 to 3, after which the decrease becomes more gradual.
Evaluating the bound at $\mu=3$ provides a simple and conservative sufficient condition that holds for all $\mu \ge 3$, yielding the inequality in Eq.~(\ref{eq:prop2_02}).
\hfill $\square$

Proposition~\ref{result:prop2}, established for the $\ket{D_{1}^{\mu}}$-based initialization, shows that incorporating a larger constraint set can improve resource efficiency under the considered oracle-cost condition.
This observation is consistent with the preprocessing in Algorithm~\ref{alg:preprocessing1}, which seeks to combine constraints into larger sets whenever possible.
Although the preprocessing strategy may not be optimal, it remains meaningful as a simple baseline for exploiting the potential advantages of larger constraint sets in a conservative setting.

The above results are established for the weight-one Dicke state $\ket{D_{1}^{\mu}}$.
For a general weight-$b$ Dicke state $\ket{D_{b}^{\mu}}$, an analogous comparison holds, but the constants appearing in the criteria depend on the search space ratios $|F_{\sigma_{i}}|/|F_{\sigma_{i+1}}|$ and $|F_{\sigma}|/|F_{\sigma'}|$, as well as the ratio $|S|/|F|$ of the number of solutions to the search space size.
A detailed analysis of the threshold behavior for general $b$, including explicit constants and circuit-level decompositions, is left for further investigation.
Furthermore, identifying effective initialization strategies across constraint sets of varying cardinality could lead to a more efficient preprocessing procedure.

\subsubsection*{Parity constraints extracted from linear constraints}
\label{subsec:results01-02}

In analogy to the discussion of cardinality constraints in the previous subsection, we examine how parity constraints affect the resource cost of the initialization.
A key difference is that a parity constraint implemented via a GHZ-type state $\ket{GHZ_{\mu,\nu}^{(X)}}$ always reduces the search space by half, independent of the constraint set size $\mu$.
As a result, increasing $\mu$ does not yield any further reduction of the search space, but instead incurs a higher circuit cost.
It therefore follows that, for parity constraints, applying constraints on smaller subsets of qubits is generally more efficient in terms of overall resource cost.

We now examine how the number of imposed parity constraints affects the overall resource requirements.
As before, we consider a sequence of strategies $\sigma_{i}$, where $\sigma_{0}$ denotes the standard initialization and each successive strategy incorporates one additional disjoint constraint set of size $\mu$.
Adding such a set via $\ket{GHZ_{\mu,\nu}^{(X)}}$ halves the search space, that is, $|F_{\sigma_i}|/|F_{\sigma_{i+1}}|=2$.
As in the case of cardinality constraints, we consider oracle constructions in which incorporating an additional parity constraint into the initialization does not increase the oracle cost.
Meanwhile, when $\ket{GHZ_{\mu,\nu}^{(X)}}$ is implemented using the circuit in Figure~\ref{Figure02}(b), then we have $\mathcal{S}_{\sigma_{i+1}} - \mathcal{S}_{\sigma_{i}} \le 2\mu+1$.
Hence, by applying the same comparison argument as in Proposition~\ref{result:prop1}, we obtain the following result.

\begin{figure*}[t]
\centering
\includegraphics[width=0.9\textwidth]{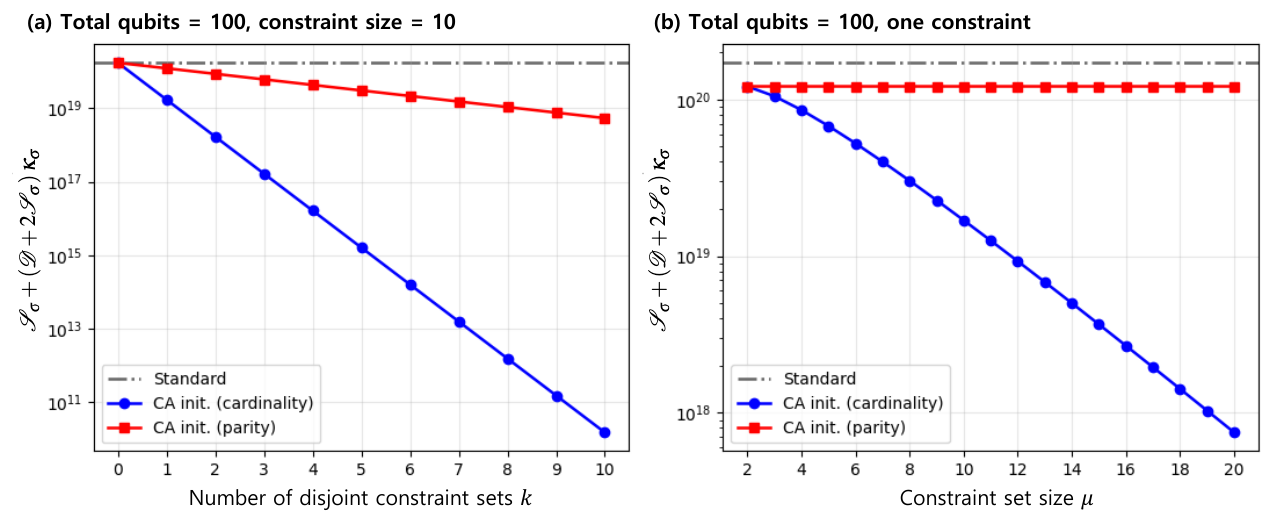}
\caption{
Resource cost of constraint-aware initialization for $n=100$ qubits with $|S|=1$.
The vertical axis shows the non-oracle resource cost $\mathcal{S}_{\sigma}+(\mathcal{D}+2\mathcal{S}_{\sigma})\kappa_{\sigma}$ measured in total gate count, where the multi-controlled $X$ gate in $\mathcal{D}$ is decomposed using the gate set $\{RZ, H, CNOT\}$.
The dash-dot line (Standard) corresponds to the uniform superposition initialization $\ket{+}^{\otimes 100}$.
(a) Resource cost as the number of incorporated disjoint constraint sets $k$ increases, with each constraint set acting on 10 qubits.
Cardinality constraints are implemented using $\ket{D_{1}^{10}}$, whereas parity constraints are implemented using $\ket{GHZ_{10, 1}^{(X)}}$.
Both approaches reduce the resource cost relative to the standard baseline, with cardinality constraints achieving substantially greater reductions than parity constraints.
(b) Resource cost as the size $\mu$ of a single incorporated constraint set increases.
Cardinality and parity constraints are implemented using $\ket{D_{1}^{\mu}}$ and $\ket{GHZ_{\mu, 1}^{(X)}}$, respectively.
The resource cost decreases rapidly with $\mu$ for cardinality constraints, while it increases slightly yet remains below the standard baseline for parity constraints. 
}
\label{Figure03}
\end{figure*}

\begin{Prop}
\label{result:prop3}
Suppose that strategy $\sigma_{i+1}$ is obtained from $\sigma_{i}$ by incorporating one additional disjoint constraint set represented by $\ket{GHZ_{\mu,\nu}^{(X)}}$, where $\mu\ge 2$.
Suppose further that $\mathcal{O}_{\sigma_{i+1}} \le \mathcal{O}_{\sigma_i}$, $|F_{\sigma_i}| \ge 64|S|$, and $|F_{\sigma_{i+1}}| \ge 2|S|$, where $|S|>0$ is the number of solutions.
If
\begin{equation}
\label{eq:prop3_01}
\mathcal{O}_{\sigma_{i}} + \mathcal{D} >\left(\frac{24\pi}{21\sqrt{2}-8\pi}\right)\left(\mathcal{S}_{\sigma_{i+1}}-\mathcal{S}_{\sigma_i}\right)-2\mathcal{S}_{\sigma_i},
\end{equation}
then strategy $\sigma_{i+1}$ is more efficient than strategy $\sigma_i$.
When preparing $\ket{GHZ_{\mu,\nu}^{(X)}}$ using the circuit in Figure~\ref{Figure02}(b), we have a simplified sufficient condition:
\begin{equation}
\label{eq:prop3_02}
\mathcal{O}_{\sigma_{i}} + \mathcal{D} >\frac{24\pi(2\mu+1)}{21\sqrt{2}-8\pi}
\end{equation}
for each of the considered resource metrics (total gate count, circuit depth, and number of two-qubit gates).
\end{Prop}

The right-hand side of the inequality in Eq.~(\ref{eq:prop3_02}) scales linearly with the constraint set size $\mu$.
Even so, when $\mu$ is small compared to the total number of qubits, the inequality can often be satisfied by accounting only for the cost $\mathcal{D}$.
For larger values of $\mu$, the condition can still be met in practice once the oracle construction cost $\mathcal{O}$ is taken into account.
Therefore, incorporating disjoint constraint sets, even with relatively large $\mu$, can still be advantageous in terms of overall resource efficiency under the considered oracle-cost condition.

To numerically illustrate the analytical resource scaling of the proposed framework, we compute the non-oracle resource cost 
\begin{equation}
\mathcal{S}_{\sigma}+(\mathcal{D}+2\mathcal{S}_{\sigma})\kappa_{\sigma}.
\end{equation}
Figure~\ref{Figure03} shows that the non-oracle resource cost decreases independently of the oracle implementation. 
Together with the reduced optimal query number and the oracle-cost relations established in Propositions~\ref{result:prop1}-\ref{result:prop3}, this implies a reduction in the overall resource cost when constraint-aware initialization is employed.
To assess the scale of the diffusion cost, we first examine $\mathcal{D}$, which corresponds to implementing the unitary $2\ket{0}\bra{0}^{\otimes n} - \mathbf{I}$.
The operator $2\ket{0}\bra{0}^{\otimes n} - \mathbf{I}$ is implemented using the Hadamard gate $H$, Pauli-$X$ and Pauli-$Z$ gates, together with a multi-controlled $X$ gate~\cite{Nielsen2010, Wong2022}.
Here, we further decompose the multi-controlled $X$ gate into one- and two-qubit gates over the gate set $\{RZ, H, CNOT\}$, which yielded a relatively compact one- and two-qubit gate decomposition among the implementations we examined using IBM's Qiskit 2.1.1. 
For example, even for $n=10$, the resulting circuit requires 713 one-qubit gates and 444 two-qubit gates with circuit depth of 825, and the cost increases rapidly as $n$ grows.
This already exceeds the thresholds appearing in Eqs.~\eqref{eq:prop1_02} and \eqref{eq:prop2_02}, indicating that the sufficient conditions of the propositions can be satisfied in practice without requiring an exceptionally costly oracle.

\begin{table*}[t]
\centering
\begin{tabular*}{\textwidth}{@{\extracolsep{\fill}}ccccc@{}}
\toprule
Component 
& One-qubit 
& Two-qubit
& Total
& Depth \\
\midrule
$\ket{+}^{\otimes 100}$
& 100
& -
& 100
& 1 \\

$\ket{D_{1}^{\mu}}^{\otimes k} \otimes \ket{+}^{\otimes 100-\mu k}$
& $100-(\mu-1)k$
& $2(\mu-1)k$ 
& $100+(\mu-1)k$
& $2\mu-1$ \\

$\ket{GHZ_{\mu, 1}^{(X)}}^{\otimes k} \otimes \ket{+}^{\otimes 100-\mu k}$
& $100+2k$ 
& $(\mu-1)k$ 
& $100+(\mu+1)k$ 
& $\mu+1$ \\

$2\ket{0}\bra{0}^{\otimes 100} - \mathbf{I}$
& $116813$ 
& $76764$ 
& $193577$
& $186045$\\

\bottomrule
\end{tabular*}
\caption{
Resource counts for representative initialization circuits and the diffusion operator.
The cost of implementing $2\ket{0}\bra{0}^{\otimes 100} - \mathbf{I}$ is estimated after decomposing the multi-controlled $X$ gate into the gate set $\{RZ, H, CNOT\}$, whereas the initialization cost follows the circuit construction in Figure~\ref{Figure02}, where each controlled rotation is counted as a single two-qubit gate.
}
\label{table:resource_count}
\end{table*}

Figure~\ref{Figure03} shows the resulting non-oracle resource cost measured in total gate count for representative cardinality- and parity-based initializations with $n=100$ and $|S|=1$.
The resource counts used in this analysis are summarized in Table~\ref{table:resource_count}, and the optimal number of queries is computed according to Eq.~\eqref{eq:optimal_query}.
In Figure~\ref{Figure03}(a), the size of each incorporated constraint set is fixed to 10 qubits, and the number of disjoint constraint sets is varied.
For cardinality and parity constraints, the initial states are prepared as $\ket{D_{1}^{10}}^{\otimes k} \otimes \ket{+}^{\otimes 100-10k}$ and $\ket{GHZ_{10, 1}^{(X)}}^{\otimes k} \otimes \ket{+}^{\otimes 100-10k}$.
Both cardinality and parity constraints reduce the resource cost relative to the standard uniform initialization $\ket{+}^{100}$, with the reduction being substantially greater for cardinality constraints.
Moreover, the resource cost decreases as additional disjoint constraint sets are incorporated, consistent with the trends predicted by Propositions~\ref{result:prop1} and \ref{result:prop3}.
In Figure~\ref{Figure03}(b), a single constraint set is incorporated and its size $\mu$ is varied.
The resource cost decreases rapidly for cardinality constraints, as predicted by Proposition~\ref{result:prop2}, whereas it slightly increases for parity constraints because the search-space reduction remains fixed at a factor of two while the state-preparation cost grows with $\mu$.
Nevertheless, the parity-based initialization remains more resource-efficient than the standard uniform initialization throughout the range considered, providing additional numerical support for the resource-efficiency trends predicted by Propositions~\ref{result:prop1}-\ref{result:prop3}.

\begin{figure}[t]
\centering
\includegraphics[width=0.6\textwidth]{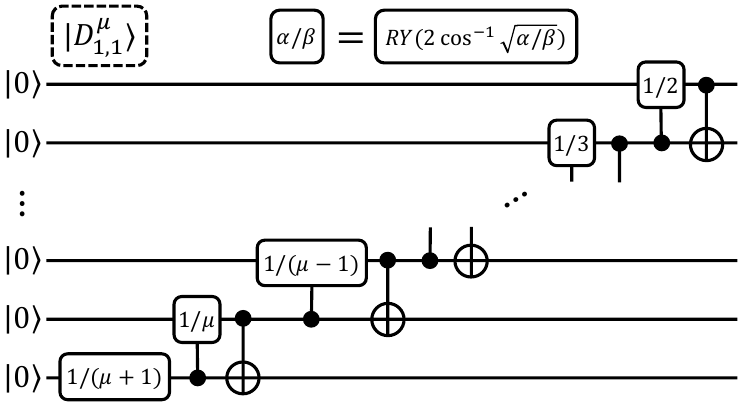}
\caption{
Quantum circuit for preparing $\ket{D_{1,1}^{\mu}}$ on $\mu$ qubits.
}
\label{Figure04}
\end{figure}

\subsection*{Simulation results for cardinality constraints}
\label{subsec:results02}

To illustrate the effectiveness of the proposed initialization scheme, we consider the exact-cover problem as a representative combinatorial problem.
Given a collection $\mathcal{A} = \{A_{1}, A_{2}, \dots, A_{n}\}$ of subsets of a universal set $U = \{u_{1}, u_{2}, \dots, u_{m}\}$, 
a subcollection $\mathcal{A}'   \subseteq  \mathcal{A}$ is said to form an exact-cover of $U$ if it satisfies the following two conditions:
\begin{itemize}
    \item[(i)] Any two distinct subsets $A_{i}, A_{j} \in \mathcal{A}'$ are disjoint:
    \begin{equation*}
        A_{i} \cap A_{j} = \emptyset.
    \end{equation*}
    
    \item[(ii)] The union of all subsets in $\mathcal{A}'$ covers the entire set $U$:
    \begin{equation*}
        \bigcup_{A_{i} \in \mathcal{A}'} A_{i} = U.
    \end{equation*}
\end{itemize}
The objective is to find such a subcollection $\mathcal{A}'$ among the $2^{n}$ possible choices.

Equivalently, the problem can be formulated in terms of a binary string $x=x_{1}x_{2}\dots x_{n} \in \{0, 1\}^{n}$, where $x_{i}=1$ indicates that subset $A_{i}$ is selected.
For each $j \in \{1, 2, ..., m\}$, we define 
\begin{equation}
\label{eq:constraint_set_exact-cover}
C_{j} = \{i~ | u_{j} \in A_{i}\},
\end{equation}
the set of indices of subsets that contain $u_{j}$.
The exact-cover condition then requires that
\begin{equation}
\label{eq:constraints_exact-cover}
\sum_{i \in C_{j}}x_{i}=1, \quad j=1, \dots, m,
\end{equation}
which enforces that each element is covered by exactly one selected subset.

Suppose that, after performing the classical preprocessing described in Algorithm~\ref{alg:preprocessing1} and reordering the qubit indices, we obtain a collection of disjoint subsets $C'_{1}, \dots, C'_{t}$ and reduced subsets $R'_{t+1}, \dots, R'_{s}$ derived from overlapping constraints.
Let $k_{j}$ denote the number of qubits in the $j$-th subset.
For each disjoint subset $C'_{j}$, we prepare the Dicke state $\ket{D_{1}^{k_{j}}}$ acting on the qubits $\mathbf{q}_{C'_{j}}$ to represent the corresponding constraint.
For each reduced subset $R'_{j}$, we employ the state $\ket{D_{1,1}^{k_{j}}}$ defined in Eq.~(\ref{eq:relaxed_Dicke}), whose explicit form is
\begin{equation}
\label{eq:linear_combi_Dicke_ecp}
\ket{D_{1,1}^{k_{j}}}_{\mathbf{q}_{R'_{j}}}=\frac{1}{\sqrt{k_{j}+1}}\ket{00 \cdots 0}_{\mathbf{q}_{R'_{j}}} + \frac{\sqrt{k_{j}}}{\sqrt{k_{j}+1}}\ket{D_{1}^{k_{j}}}_{\mathbf{q}_{R'_{j}}}.
\end{equation}
The state $\ket{D_{1,1}^{k_{j}}}$ can be prepared using the circuit shown in Figure~\ref{Figure04}.



\begin{figure}[t]
\centering
\includegraphics[width=0.6\textwidth]{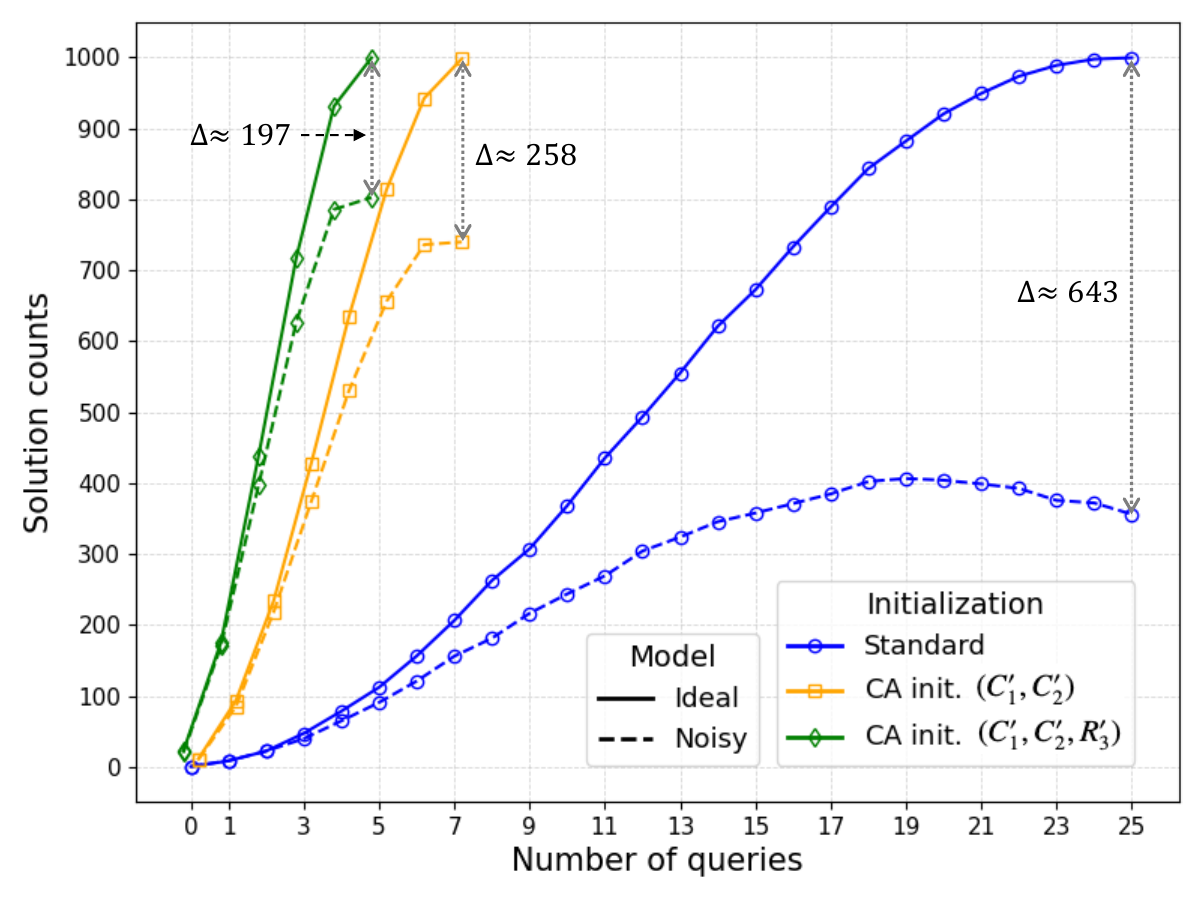}
\caption{
Comparison of the solution counts for the exact-cover instance in Eq.~(\ref{ex:exact-cover}) under different initialization strategies.
The standard uniform initialization is compared with two constraint-aware initialization~(CA init.) schemes incorporating $(C'_{1}, C'_{2})$ and $(C'_{1}, C'_{2}, R'_{3})$, respectively.
Solid (dashed) curves correspond to the ideal (noisy) model, where depolarizing noise with error rates of $10^{-5}$ and $10^{-4}$ is applied to one- and two-qubit gates.
Each data point is obtained by averaging over 20 independent circuit executions, each consisting of 1000 measurement shots.
The standard deviation across runs is at most 21 counts.
Bidirectional arrows indicate the difference in the solution counts between the ideal and noisy models, evaluated at the optimal number of queries for each initialization scheme.
}
\label{Figure05}
\end{figure}

We now look at the following example of the exact-cover problem.
Let $U=\{1, 2, 3, 4, 5, 6, 7\}$ and the collection of subsets be given as
\begin{align}
\label{ex:exact-cover}
&A_{1}=\{1, 5\}, ~~ A_{2}=\{1, 3, 6\}, ~~ A_{3}=\{1, 2, 5\}, ~~ A_{4}=\{1, 7\}, ~~ A_{5}=\{3, 4\}, \nonumber\\
&A_{6}=\{4, 6\}, ~~ A_{7}=\{2, 4, 5\},
~~ A_{8}=\{2, 7\}, ~~ A_{9}=\{6, 7\}, ~~ A_{10}=\{3, 5, 7\}.
\end{align}
This instance admits a unique exact-cover, given by $\mathcal{A}'=\{A_{3}, A_{5}, A_{9}\}$, so that $|S|$=1.
From the subsets in Eq.~(\ref{ex:exact-cover}), the corresponding constraint sets are constructed according to Eq.~(\ref{eq:constraint_set_exact-cover}) as
\begin{align}
&C_{1}=\{1, 2, 3, 4\}, ~~ C_{2}=\{3, 7, 8\}, ~~ C_{3}=\{2, 5, 10\}, ~~ C_{4}=\{5, 6, 7\}, \nonumber\\ 
&C_{5}=\{1, 3, 7, 10\}, ~~ C_{6}=\{2, 6, 9\}, ~~ C_{7}=\{4, 8, 9, 10\}.
\end{align}
Applying the classical preprocessing procedure described in Algorithm~\ref{alg:preprocessing1}, we first identify disjoint sets that will be used to construct the circuit for initialization.
When the threshold parameter $\eta_{\rm{th}}=0$, only strictly disjoint sets are allowed, and the preprocessing selects
\begin{equation}
C'_{1}=C_{1},~~C'_{2}=C_{4}.
\end{equation}
When the threshold is relaxed to $\eta_{\rm{th}}=1$, a small degree of overlap is permitted, which allows an additional set to be obtained as
\begin{equation}
R'_{3}=C_{7}\setminus (C_{1}\cup C_{4})=\{8, 9, 10\}.
\end{equation}
The initial quantum state is then constructed by preparing Dicke states on the subsystems corresponding to $C'_{1}$ and $C'_{2}$, while the subsystem associated with $R'_{3}$ is prepared according to the quantum state of the form in Eq.~(\ref{eq:linear_combi_Dicke_ecp}).
For comparison, the initialization strategies considered in this example, together with the corresponding initial states, search-space sizes, and optimal number of queries, are summarized in Table~\ref{table:cardinality_example}.

We perform simulations using IBM's Qiskit, an open-source software development kit for quantum circuit design and simulation~\cite{Qiskit2024}.
The oracle operator is assumed to be ideal to reduce simulation overhead and to focus on the resource impact of the initialization and diffusion components.
We note, however, that oracles for the exact-cover problem can be explicitly constructed using existing approaches, such as those proposed by Gilliam \textit{et al.}~\cite{Gilliam2021} and Jiang \textit{et al.}~\cite{Jiang2023}.
When simulating the noisy model, the multi-controlled $X$ gate in the operator $2\ket{0}\bra{0}^{\otimes n} - \mathbf{I}$ is decomposed into one- and two-qubit gates over the gate set $\{RZ, H, CNOT\}$, consistent with the implementation used in the resource analysis.
To assess the effect of noise on the performance of different initialization strategies, we consider a depolarizing noise model applied to all operations except for the oracle operator.
Specifically, depolarizing noise with error rates of $10^{-5}$ and $10^{-4}$ is applied to one- and two-qubit gates, respectively.
For each value of the query number, the circuit is executed 20 times, each run consisting of $1000$ measurement shots.
The reported results correspond to the average number of measurements yielding a solution state.

Figure~\ref{Figure05} shows the solution counts as a function of the number of queries for the exact-cover instance defined in Eq.~(\ref{ex:exact-cover}), obtained using the standard uniform initialization and two constraint-aware initialization strategies.
The first constraint-aware strategy incorporates only the disjoint constraint sets $(C'_{1}, C'_{2})$, while the second additionally includes the reduced set $R'_{3}$ derived from overlapping constraints.
For each initialization strategy, results are shown for both the ideal and noisy models.
In the standard initialization, the corresponding optimal number of queries is $25$, whereas the first constraint-aware strategy reduces the optimal number of queries to $7$, and the second strategy further reduces it to $5$.
Although the constraint-aware initialization increases the circuit complexity associated with initial state preparation and the diffusion operator, the substantial reduction in the optimal number of queries leads to improved robustness against noise at the respective optimal query numbers.

\begin{figure}[t]
\centering
\includegraphics[width=0.6\textwidth]{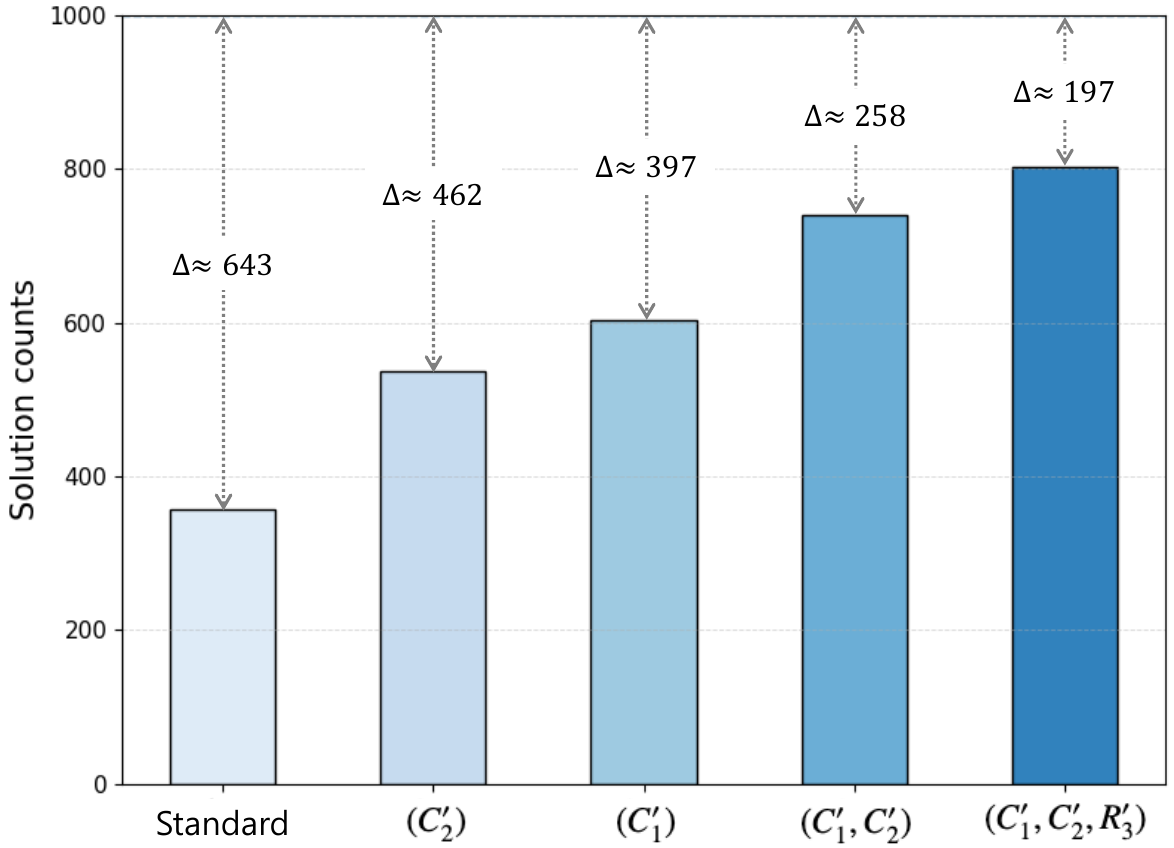}
\caption{
Difference in solution counts between the ideal and noisy models, evaluated at the respective optimal number of queries for different initialization strategies, where the bars show the solution counts in the noisy model and $\Delta$ denotes the corresponding difference.
All simulation settings are the same as in Figure~\ref{Figure05}, and the standard deviation across runs is at most 19 counts.
The strategies include the standard uniform initialization and several constraint-aware initializations incorporating single constraint sets $(C'_{2})$ and $(C'_{1})$, as well as multiple constraint sets obtained through preprocessing, namely $(C'_{1}, C'_{2})$ and $(C'_{1}, C'_{2}, R'_{3})$.
}
\label{Figure06}
\end{figure}

\begin{table*}[t]
\centering
\begin{tabular*}{\textwidth}{@{\extracolsep{\fill}}cccc@{}}
\toprule
Strategy 
& Initial state 
& Search space size 
& Optimal query \\
\midrule
Standard 
& $\ket{+}^{\otimes 10}$ 
& $2^{10}=1024$ 
& $25$ \\

CA init. ($C'_2$)
& $\ket{D_{1}^{3}}\otimes \ket{+}^{\otimes 7}$ 
& $3\cdot 2^7=384$ 
& $15$ \\

CA init. ($C'_1$)
& $\ket{D_{1}^{4}}\otimes \ket{+}^{\otimes 6}$ 
& $4\cdot 2^6=256$ 
& $12$ \\

CA init. ($C'_1, C'_2$)
& $\ket{D_{1}^{4}}\otimes \ket{D_{1}^{3}}\otimes \ket{+}^{\otimes 3}$ 
& $4\cdot 3\cdot 2^3=96$ 
& $7$ \\

CA init. ($C'_1, C'_2, R'_3$)
& $\ket{D_{1}^{4}}\otimes \ket{D_{1}^{3}}\otimes \ket{D_{1,1}^{3}}$ 
& $4\cdot 3\cdot 4=48$ 
& $5$ \\
\bottomrule
\end{tabular*}
\caption{
Summary of the initialization strategies used in the cardinality-constraint example.
Here, CA init. ($\cdot$) denotes the constraint-aware initialization obtained by incorporating the indicated constraint sets into the initial state preparation.
The strategies CA init. ($C'_1, C'_2$) and CA init. ($C'_1, C'_2, R'_3$) are obtained via the classical preprocessing in Algorithm~\ref{alg:preprocessing1}.
}
\label{table:cardinality_example}
\end{table*}

Figure~\ref{Figure06} extends the comparison by examining the difference in solution counts between the ideal and noisy models at the respective optimal number of queries for a broader set of initialization strategies.
In addition to the standard uniform initialization and the two constraint-aware schemes based on $(C'_{1}, C'_{2})$ and $(C'_{1}, C'_{2}, R'_{3})$, which were discussed earlier, we also include constraint-aware initializations incorporating only a single constraint set, namely $(C'_{2})$ and $(C'_{1})$.
These two cases correspond to reduced search spaces of size 3/8 and 1/4 of the reference space of size $2^{10}$, respectively, which result in optimal query numbers of 15 and 12.
As shown in Figure~\ref{Figure06}, the solution counts obtained in the noisy model increase as we move from the standard initialization to constraint-aware strategies that incorporate larger constraint sets and additional disjoint constraint sets.
These results provide supporting evidence for the constraint selection strategy adopted in Algorithm~\ref{alg:preprocessing1}, where larger constraint sets are prioritized and the procedure is designed to identify as many disjoint constraint sets as possible.

Importantly, Figure~\ref{Figure06} shows that even in unfavorable preprocessing scenarios where only a single constraint set is identified, the corresponding constraint-aware initialization still yields a noticeable improvement over the standard uniform initialization.
This suggests that the proposed preprocessing-guided initialization offers a systematic and practically feasible way to incorporate constraints.
While incorporating additional constraints, including non-disjoint ones, may further improve performance, it is nontrivial to design a systematic initialization that directly incorporates non-disjoint constraint sets.
Our preprocessing instead relaxes overlaps by extracting reduced subsets that can be implemented in the initialization, and the simulation results support the effectiveness of this approach.
Moreover, identifying a large number of disjoint constraint sets is itself challenging.
Accordingly, we adopt a greedy preprocessing strategy that does not aim to be optimal but provides a simple and implementable baseline that already delivers a clear performance advantage in our simulations.

\subsection*{Simulation results for parity constraints extracted from linear constraints}
\label{subsec:results03}

We now introduce a variant of the exact-cover problem to illustrate the incorporation of parity constraints extracted from linear constraints through GHZ-type initialization.
The universal set $U$ and the collection of subsets $\{A_{i}\}$ are taken to be the same as in the previous subsection.
In contrast to the standard exact-cover formulation, we do not require the selected subsets to be disjoint. 
Instead, we require each element of $U$ to be covered exactly twice. 
We assign weights $\alpha_{i}$ to the subsets such that when $A_{i}$ is selected, it contributes $\alpha_{i}$ to the coverage count of each element it contains. 
We set $\alpha_{i}=2$ for $i \in \{3, 6, 10\}$ and $\alpha_{i}=1$ otherwise.
Introducing binary variables $x_{i} \in \{0, 1\}$, this condition is expressed as
\begin{equation}
\label{eq:constraints_exact-cover_weight}
\sum_{i \in C_{j}}\alpha_{i}x_{i}=2, \quad j=1, \dots, 7,
\end{equation}
and this instance admits a unique solution given by $\mathcal{A}'=\{A_{1}, A_{2}, A_{5}, A_{7}, A_{8}, A_{9}\}$.

Using the definition of $\hat{C}_{j}=\{i\in C_{j}|\alpha_{i}~\text{is odd}\}$, we obtain
\begin{align}
\hat{C}_{1}=\{1, 2, 4\}, ~~ \hat{C}_{2}=\{7, 8\}, ~~ \hat{C}_{3}=\{2, 5\}, ~~ \hat{C}_{4}=\{5, 7\}, ~~ \hat{C}_{5}=\{1, 7\}, ~~ \hat{C}_{6}=\{2, 9\}, ~~ \hat{C}_{7}=\{4, 8, 9\}.
\end{align}
Accordingly, by applying the preprocessing step in Algorithm~\ref{alg:preprocessing2}, we identify disjoint sets $C'_{1}=\hat{C}_{2}$ and $C'_{2}=\hat{C}_{3}$.
We then implement the GHZ-type initialization on the qubit registers associated with these constraint sets.
Since two disjoint parity constraint sets are incorporated, the effective search space is reduced by a factor of $2^{2}=4$, which yields an optimal number of queries of 12 for this instance.
Table~\ref{table:parity_example} summarizes the constraint sets encoded by the proposed initialization strategy for this example, together with the corresponding initial states, search-space sizes, and optimal number of queries.

\begin{figure}[t]
\centering
\includegraphics[width=0.6\textwidth]{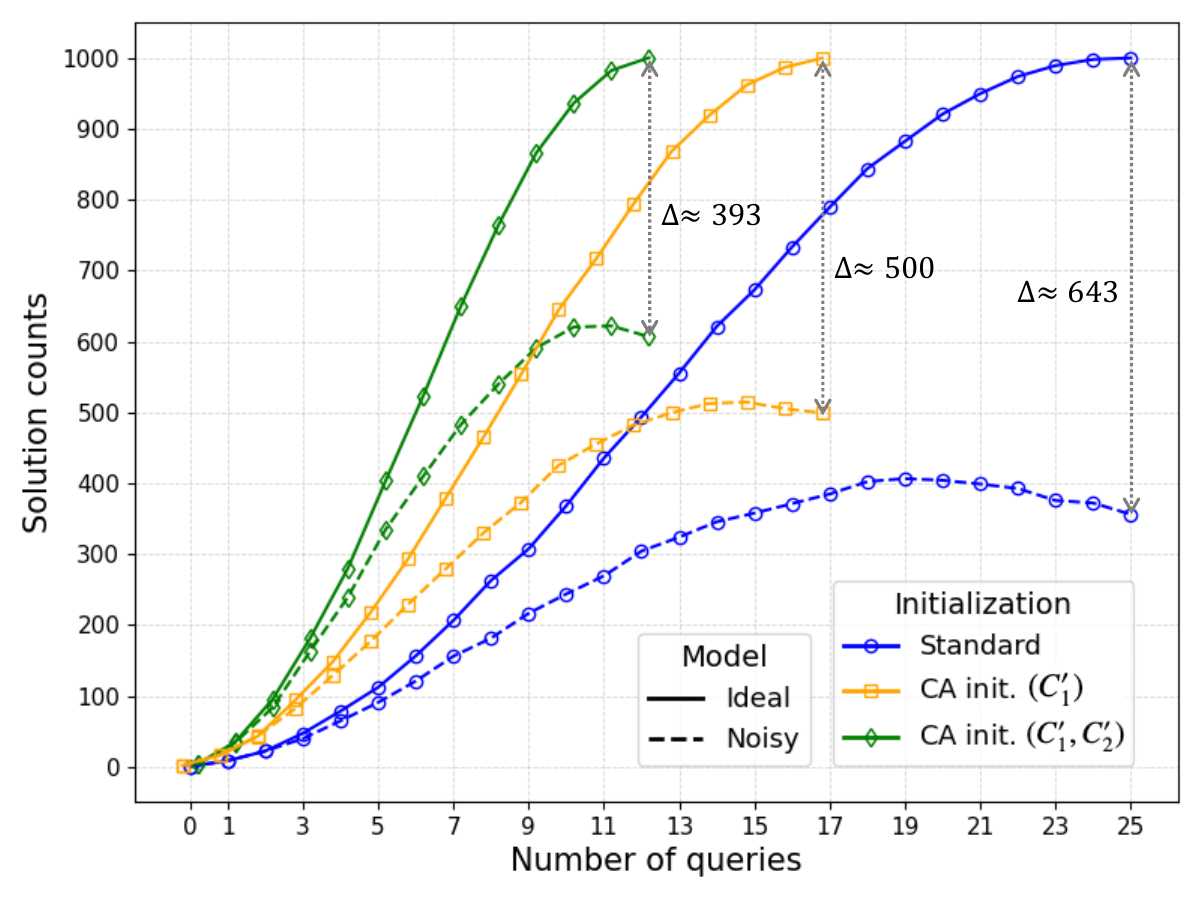}
\caption{
Comparison of the solution counts under different initialization strategies for parity constraints extracted from linear constraints.
The standard uniform initialization is compared with two GHZ-type initializations incorporating $(C'_{1})$ and $(C'_{1}, C'_{2})$, respectively.
Consistent with the cardinality case, incorporating additional parity constraint information improves performance under noise, and even incorporating a single constraint set yields a clear advantage over the uniform initialization.
The simulation setting is the same as in Figure~\ref{Figure05}, and the standard deviation across runs is at most 21 counts.
}
\label{Figure07}
\end{figure}

\begin{table*}[t]
\centering
\begin{tabular*}{\textwidth}{@{\extracolsep{\fill}}cccc@{}}
\toprule
Strategy 
& Initial state 
& Search space size 
& Optimal query \\
\midrule
Standard 
& $\ket{+}^{\otimes 10}$ 
& $2^{10}=1024$ 
& $25$ \\

CA init. ($C'_1$)
& $\ket{GHZ^{(X)}_{2,0}}\otimes \ket{+}^{\otimes 8}$ 
& $2^9=512$ 
& $17$ \\

CA init. ($C'_1, C'_2$)
& $\ket{GHZ^{(X)}_{2,0}}\otimes \ket{GHZ^{(X)}_{2,0}}\otimes \ket{+}^{\otimes 6}$ 
& $2^8=256$ 
& $12$ \\
\bottomrule
\end{tabular*}
\caption{
Summary of the initialization strategies used in the parity-constraints example.
The strategy CA init. ($C'_1, C'_2$) is obtained via the classical preprocessing in Algorithm~\ref{alg:preprocessing2}.
}
\label{table:parity_example}
\end{table*}

Figure~\ref{Figure07} presents the simulation results obtained under the same settings as in the previous subsection.
We plot the solution counts as a function of the number of queries for the ideal and noisy cases.
The results compare three initialization strategies that include applying GHZ-type initialization to both $C'_{1}$ and $C'_{2}$, applying it only to $C'_{1}$, and the standard uniform initialization.
As in the cardinality case, incorporating additional constraint information leads to improved performance under noise.
Even incorporating a single parity constraint set already yields a clear improvement over the standard uniform initialization.
Together with the results in the previous subsection, this confirms that the proposed constraint-aware initialization framework can systematically exploit different types of linear constraints to improve practical performance.

\section*{Conclusion}
\label{sec:conclusion}

In this work, we presented a systematic framework for constraint-aware initialization in Grover's algorithm applied to combinatorial problems with linear constraints, which reduces the effective search space by incorporating constraint information already at the initialization step.
The proposed approach combines a simple classical preprocessing step with structured initial state preparation.
The preprocessing selects collections of mutually disjoint constraint sets that are jointly implementable in the initialization.
For each selected set, we prepare a structured quantum state that encodes the corresponding constraint information.
We use Dicke states to encode cardinality constraints and GHZ-type states to incorporate parity constraints extracted from more general linear constraints.
The overall initial state is then constructed as a tensor product over the selected sets.
We also provided a conservative circuit-level resource analysis, showing that although structured initialization introduces additional preparation costs, it can still improve overall resource efficiency in terms of gate counts and circuit depth compared to the standard uniform initialization.
The practical effectiveness of the framework was further supported by numerical simulations of the exact-cover problem, demonstrating clear performance advantages, particularly in the presence of noise.

Several important directions remain for future work.
While the present analysis focused on specific instances of Dicke states and GHZ-type states, more general Dicke states are known to admit polynomial-time circuit constructions~\cite{Bartschi2019}, suggesting that similar resource-efficiency trends may extend beyond the cases considered here.
In addition, the circuit implementations examined in this work are not optimized for minimal resource usage, and further improvements may be achieved through more efficient state preparation schemes.
It would also be interesting to evaluate the proposed framework on real quantum hardware, where the impact of the platform-specific gate sets and noise characteristics on its practical resource efficiency could be assessed.
The classical preprocessing procedure introduced in this framework is likewise intended as a simple and practical baseline rather than an optimal approach, and identifying more effective preprocessing methods will be important for achieving both practical efficiency and near-optimal performance.

Another natural direction is to extend the framework beyond equality-type linear constraints to inequality constraints.
Such constraints may be incorporated by preparing superpositions of Dicke states across a range of Hamming weights, which suggests that the proposed approach could be generalized to a broader class of combinatorial problems.
It would also be worthwhile to compare the proposed approach with divide-and-conquer strategies, in which the problem is decomposed into subproblems each solved by an independent algorithm.
While such strategies reduce the per-instance resource requirements, they require either multiple quantum processors or sequential executions, whereas the proposed framework addresses the original problem within a single Grover search instance by incorporating available constraint information into the initialization.
A concrete analysis of this trade-off remains an interesting open question.
Finally, it would be interesting to apply the proposed framework within optimization settings such as Grover adaptive search~\cite{Gilliam2021} and variational algorithms with the Grover mixer quantum alternating operator ansatz~\cite{Bartschi2020}, and to evaluate its impact on practically relevant combinatorial optimization problems.
A thorough comparison of the proposed framework with related constraint-aware approaches in terms of resource efficiency and practical performance would also provide valuable insight.

\section*{Funding}
This work was supported by the Institute for Information \& Communications Technology Promotion (IITP) (2019-0-00003) and the Global TOP strategic Research Program of the National Research Council of Science \& Technology (NST) (GTL25011-000) grant funded by the Korea government (MSIT). 
M.C. acknowledges support from the Korea Institute of Science and Technology Information (KISTI) (Grant No. K26L3M3C3).
E.B. is supported by the National Research Foundation of Korea (NRF) of Korea grant funded by the Ministry of Science and ICT (MSIT) (Grant No. RS-2025-07742969).

\section*{Code availability}
The simulation code used in this work is publicly available at https://github.com/mathcmj/grover-ca-init and has been archived on Zenodo at https://doi.org/10.5281/zenodo.21385709.

\bibliography{reference}

\begin{thebibliography}{10}
\urlstyle{rm}
\expandafter\ifx\csname url\endcsname\relax
  \def\url#1{\texttt{#1}}\fi
\expandafter\ifx\csname urlprefix\endcsname\relax\def\urlprefix{URL }\fi
\expandafter\ifx\csname doiprefix\endcsname\relax\def\doiprefix{DOI: }\fi
\providecommand{\bibinfo}[2]{#2}
\providecommand{\eprint}[2][]{\url{#2}}

\bibitem{Grover1996}
\bibinfo{author}{Grover, L.~K.}
\newblock \bibinfo{title}{A fast quantum mechanical algorithm for database search}.
\newblock In \emph{\bibinfo{booktitle}{Proceedings of the twenty-eighth annual ACM symposium on Theory of computing}}, \bibinfo{pages}{212--219} (\bibinfo{year}{1996}).

\bibitem{Brassard2000}
\bibinfo{author}{Brassard, G.}, \bibinfo{author}{Hoyer, P.}, \bibinfo{author}{Mosca, M.} \& \bibinfo{author}{Tapp, A.}
\newblock \bibinfo{journal}{\bibinfo{title}{Quantum amplitude amplification and estimation}}.
\newblock {\emph{\JournalTitle{arXiv preprint quant-ph/0005055}}}  (\bibinfo{year}{2000}).

\bibitem{Montanaro2016}
\bibinfo{author}{Montanaro, A.}
\newblock \bibinfo{journal}{\bibinfo{title}{Quantum algorithms: an overview}}.
\newblock {\emph{\JournalTitle{npj Quantum Information}}} \textbf{\bibinfo{volume}{2}}, \bibinfo{pages}{15023} (\bibinfo{year}{2016}).

\bibitem{Brassard1998}
\bibinfo{author}{Brassard, G.}, \bibinfo{author}{H{\o}yer, P.} \& \bibinfo{author}{Tapp, A.}
\newblock \bibinfo{title}{Quantum counting}.
\newblock In \emph{\bibinfo{booktitle}{International Colloquium on Automata, Languages, and Programming}}, \bibinfo{pages}{820--831} (\bibinfo{organization}{Springer}, \bibinfo{year}{1998}).

\bibitem{Aaronson2020}
\bibinfo{author}{Aaronson, S.} \& \bibinfo{author}{Rall, P.}
\newblock \bibinfo{title}{Quantum approximate counting, simplified}.
\newblock In \emph{\bibinfo{booktitle}{Symposium on simplicity in algorithms}}, \bibinfo{pages}{24--32} (\bibinfo{organization}{SIAM}, \bibinfo{year}{2020}).

\bibitem{Durr1996}
\bibinfo{author}{Durr, C.} \& \bibinfo{author}{H{\o}yer, P.}
\newblock \bibinfo{journal}{\bibinfo{title}{A quantum algorithm for finding the minimum}}.
\newblock {\emph{\JournalTitle{arXiv preprint arXiv:quant-ph/9607014}}}  (\bibinfo{year}{1996}).

\bibitem{Baritompa2005}
\bibinfo{author}{Baritompa, W.~P.}, \bibinfo{author}{Bulger, D.~W.} \& \bibinfo{author}{Wood, G.~R.}
\newblock \bibinfo{journal}{\bibinfo{title}{Grover's quantum algorithm applied to global optimization}}.
\newblock {\emph{\JournalTitle{SIAM Journal on Optimization}}} \textbf{\bibinfo{volume}{15}}, \bibinfo{pages}{1170--1184} (\bibinfo{year}{2005}).

\bibitem{Gilliam2021}
\bibinfo{author}{Gilliam, A.}, \bibinfo{author}{Woerner, S.} \& \bibinfo{author}{Gonciulea, C.}
\newblock \bibinfo{journal}{\bibinfo{title}{Grover adaptive search for constrained polynomial binary optimization}}.
\newblock {\emph{\JournalTitle{Quantum}}} \textbf{\bibinfo{volume}{5}}, \bibinfo{pages}{428} (\bibinfo{year}{2021}).

\bibitem{Nigatu2025}
\bibinfo{author}{Nigatu, H.} \emph{et~al.}
\newblock \bibinfo{journal}{\bibinfo{title}{Quantum machine learning and grover's algorithm for quantum optimization of robotic manipulators}}.
\newblock {\emph{\JournalTitle{IEEE Robotics and Automation Letters}}} \textbf{\bibinfo{volume}{10}}, \bibinfo{pages}{13090--13097} (\bibinfo{year}{2025}).

\bibitem{Campbell2019}
\bibinfo{author}{Campbell, E.}, \bibinfo{author}{Khurana, A.} \& \bibinfo{author}{Montanaro, A.}
\newblock \bibinfo{journal}{\bibinfo{title}{Applying quantum algorithms to constraint satisfaction problems}}.
\newblock {\emph{\JournalTitle{Quantum}}} \textbf{\bibinfo{volume}{3}}, \bibinfo{pages}{167} (\bibinfo{year}{2019}).

\bibitem{Ohno2024}
\bibinfo{author}{Ohno, H.}
\newblock \bibinfo{journal}{\bibinfo{title}{Grover’s search with learning oracle for constrained binary optimization problems}}.
\newblock {\emph{\JournalTitle{Quantum Machine Intelligence}}} \textbf{\bibinfo{volume}{6}}, \bibinfo{pages}{12} (\bibinfo{year}{2024}).

\bibitem{Sano2024}
\bibinfo{author}{Sano, Y.}, \bibinfo{author}{Mitarai, K.}, \bibinfo{author}{Yamamoto, N.} \& \bibinfo{author}{Ishikawa, N.}
\newblock \bibinfo{journal}{\bibinfo{title}{Accelerating grover adaptive search: Qubit and gate count reduction strategies with higher order formulations}}.
\newblock {\emph{\JournalTitle{IEEE Transactions on Quantum Engineering}}} \textbf{\bibinfo{volume}{5}}, \bibinfo{pages}{1--12} (\bibinfo{year}{2024}).

\bibitem{Hadfield2019}
\bibinfo{author}{Hadfield, S.} \emph{et~al.}
\newblock \bibinfo{journal}{\bibinfo{title}{From the quantum approximate optimization algorithm to a quantum alternating operator ansatz}}.
\newblock {\emph{\JournalTitle{Algorithms}}} \textbf{\bibinfo{volume}{12}}, \bibinfo{pages}{34} (\bibinfo{year}{2019}).

\bibitem{Bartschi2020}
\bibinfo{author}{B{\"a}rtschi, A.} \& \bibinfo{author}{Eidenbenz, S.}
\newblock \bibinfo{title}{Grover mixers for qaoa: Shifting complexity from mixer design to state preparation}.
\newblock In \emph{\bibinfo{booktitle}{2020 IEEE International Conference on Quantum Computing and Engineering (QCE)}}, \bibinfo{pages}{72--82} (\bibinfo{organization}{IEEE}, \bibinfo{year}{2020}).

\bibitem{Vikstaal2020}
\bibinfo{author}{Vikst{\aa}l, P.} \emph{et~al.}
\newblock \bibinfo{journal}{\bibinfo{title}{Applying the quantum approximate optimization algorithm to the tail-assignment problem}}.
\newblock {\emph{\JournalTitle{Physical Review Applied}}} \textbf{\bibinfo{volume}{14}}, \bibinfo{pages}{034009} (\bibinfo{year}{2020}).

\bibitem{Zhang2025}
\bibinfo{author}{Zhang, Z.} \emph{et~al.}
\newblock \bibinfo{journal}{\bibinfo{title}{Grover-qaoa for 3-sat: quadratic speedup, fair-sampling, and parameter clustering}}.
\newblock {\emph{\JournalTitle{Quantum Science and Technology}}} \textbf{\bibinfo{volume}{10}}, \bibinfo{pages}{015022} (\bibinfo{year}{2025}).

\bibitem{Bengtsson2020}
\bibinfo{author}{Bengtsson, A.} \emph{et~al.}
\newblock \bibinfo{journal}{\bibinfo{title}{Improved success probability with greater circuit depth for the quantum approximate optimization algorithm}}.
\newblock {\emph{\JournalTitle{Physical Review Applied}}} \textbf{\bibinfo{volume}{14}}, \bibinfo{pages}{034010} (\bibinfo{year}{2020}).

\bibitem{Metwalli2020}
\bibinfo{author}{Metwalli, S.~A.}, \bibinfo{author}{Le~Gall, F.} \& \bibinfo{author}{Van~Meter, R.}
\newblock \bibinfo{journal}{\bibinfo{title}{Finding small and large $ k $-clique instances on a quantum computer}}.
\newblock {\emph{\JournalTitle{IEEE Transactions on Quantum Engineering}}} \textbf{\bibinfo{volume}{1}}, \bibinfo{pages}{1--11} (\bibinfo{year}{2020}).

\bibitem{Mikuriya2024}
\bibinfo{author}{Mikuriya, T.}, \bibinfo{author}{Yukiyoshi, K.}, \bibinfo{author}{Fujiwara, S.}, \bibinfo{author}{de~Abreu, G. T.~F.} \& \bibinfo{author}{Ishikawa, N.}
\newblock \bibinfo{journal}{\bibinfo{title}{Grover adaptive search for the higher-order formulation of quadratic assignment problems}}.
\newblock {\emph{\JournalTitle{arXiv preprint arXiv:2410.12181}}}  (\bibinfo{year}{2024}).

\bibitem{Griset2025}
\bibinfo{author}{Griset, R.}, \bibinfo{author}{Lavdas, I.} \& \bibinfo{author}{Guth~Jarkovsk{\'y}, J.}
\newblock \bibinfo{journal}{\bibinfo{title}{State-space reduction techniques exploiting specific constraints for quantum search initialization, application to an outage planning problem}}.
\newblock {\emph{\JournalTitle{Journal of Quantum Computing}}} \textbf{\bibinfo{volume}{7}}, \bibinfo{pages}{81--105} (\bibinfo{year}{2025}).

\bibitem{Wilkening2025}
\bibinfo{author}{Wilkening, S.}, \bibinfo{author}{Ziegler, T.} \& \bibinfo{author}{Hess, M.}
\newblock \bibinfo{journal}{\bibinfo{title}{Constraint-oriented biased quantum search for linear constrained combinatorial optimization problems}}.
\newblock {\emph{\JournalTitle{arXiv preprint arXiv:2512.05205}}}  (\bibinfo{year}{2025}).

\bibitem{Cook2020}
\bibinfo{author}{Cook, J.}, \bibinfo{author}{Eidenbenz, S.} \& \bibinfo{author}{Bärtschi, A.}
\newblock \bibinfo{title}{The quantum alternating operator ansatz on maximum k-vertex cover}.
\newblock In \emph{\bibinfo{booktitle}{2020 IEEE International Conference on Quantum Computing and Engineering (QCE)}}, \bibinfo{pages}{83--92} (\bibinfo{year}{2020}).

\bibitem{Christiansen2025}
\bibinfo{author}{Christiansen, P.}, \bibinfo{author}{Binkowski, L.}, \bibinfo{author}{Ramacciotti, D.} \& \bibinfo{author}{Wilkening, S.}
\newblock \bibinfo{title}{Quantum tree generator improves qaoa state-of-the-art for the knapsack problem}.
\newblock In \emph{\bibinfo{booktitle}{Proceedings of the IEEE International Conference on Quantum Computing and Engineering (QCE 2025)}}, \bibinfo{pages}{1--10} (\bibinfo{year}{2025}).

\bibitem{Dicke1954}
\bibinfo{author}{Dicke, R.~H.}
\newblock \bibinfo{journal}{\bibinfo{title}{Coherence in spontaneous radiation processes}}.
\newblock {\emph{\JournalTitle{Physical Review}}} \textbf{\bibinfo{volume}{93}}, \bibinfo{pages}{99} (\bibinfo{year}{1954}).

\bibitem{Greenberger1989}
\bibinfo{author}{Greenberger, D.~M.}, \bibinfo{author}{Horne, M.~A.} \& \bibinfo{author}{Zeilinger, A.}
\newblock \bibinfo{title}{Going beyond bell’s theorem}.
\newblock In \emph{\bibinfo{booktitle}{Bell’s theorem, quantum theory and conceptions of the universe}}, \bibinfo{pages}{69--72} (\bibinfo{publisher}{Springer}, \bibinfo{year}{1989}).

\bibitem{Jiang2023}
\bibinfo{author}{Jiang, J.-R.} \& \bibinfo{author}{Wang, Y.-J.}
\newblock \bibinfo{title}{Quantum circuit based on grover’s algorithm to solve exact cover problem}.
\newblock In \emph{\bibinfo{booktitle}{2023 VTS Asia Pacific Wireless Communications Symposium (APWCS)}}, \bibinfo{pages}{1--5} (\bibinfo{organization}{IEEE}, \bibinfo{year}{2023}).

\bibitem{Sato2023}
\bibinfo{author}{Sato, R.}, \bibinfo{author}{Saito, K.}, \bibinfo{author}{Nikuni, T.} \& \bibinfo{author}{Watabe, S.}
\newblock \bibinfo{title}{Embedding all feasible solutions of traveling salesman problem by divide-and-conquer quantum search}.
\newblock In \emph{\bibinfo{booktitle}{2023 IEEE International Conference on Quantum Computing and Engineering (QCE)}}, vol.~\bibinfo{volume}{2}, \bibinfo{pages}{270--271} (\bibinfo{organization}{IEEE}, \bibinfo{year}{2023}).

\bibitem{Sato2025}
\bibinfo{author}{Sato, R.} \emph{et~al.}
\newblock \bibinfo{journal}{\bibinfo{title}{Two-step quantum search algorithm for solving traveling salesman problems}}.
\newblock {\emph{\JournalTitle{IEEE Transactions on Quantum Engineering}}} \textbf{\bibinfo{volume}{6}}, \bibinfo{pages}{1--12} (\bibinfo{year}{2025}).

\bibitem{Bartschi2019}
\bibinfo{author}{B{\"a}rtschi, A.} \& \bibinfo{author}{Eidenbenz, S.}
\newblock \bibinfo{title}{Deterministic preparation of dicke states}.
\newblock In \emph{\bibinfo{booktitle}{International Symposium on Fundamentals of Computation Theory}}, \bibinfo{pages}{126--139} (\bibinfo{organization}{Springer}, \bibinfo{year}{2019}).

\bibitem{Nielsen2010}
\bibinfo{author}{Nielsen, M.~A.} \& \bibinfo{author}{Chuang, I.~L.}
\newblock \emph{\bibinfo{title}{Quantum Computation and Quantum Information: 10th Anniversary Edition}} (\bibinfo{publisher}{Cambridge University Press}, \bibinfo{year}{2010}).

\bibitem{Maslov2016}
\bibinfo{author}{Maslov, D.}
\newblock \bibinfo{journal}{\bibinfo{title}{Advantages of using relative-phase toffoli gates with an application to multiple control toffoli optimization}}.
\newblock {\emph{\JournalTitle{Physical Review A}}} \textbf{\bibinfo{volume}{93}}, \bibinfo{pages}{022311} (\bibinfo{year}{2016}).

\bibitem{Wong2022}
\bibinfo{author}{Wong, T.~G.}
\newblock \emph{\bibinfo{title}{Introduction to classical and quantum computing}} (\bibinfo{publisher}{Rooted Grove Omaha, NE, USA}, \bibinfo{year}{2022}).

\bibitem{Qiskit2024}
\bibinfo{author}{Javadi-Abhari, A.} \emph{et~al.}
\newblock \bibinfo{journal}{\bibinfo{title}{Quantum computing with qiskit}}.
\newblock {\emph{\JournalTitle{arXiv preprint arXiv:2405.08810}}}  (\bibinfo{year}{2024}).

\end{thebibliography}

\section*{Author contributions statement}
E.B. and M.C. conceived the idea and performed the calculations.
M.C. carried out the simulations.
All authors discussed the results and contributed to the preparation of the manuscript.

\section*{Competing interests}
The authors declare no competing interests.

\section*{Additional information}
Correspondence and requests for materials should be addressed to M.C.

\end{document}